\documentclass[11pt, letterpaper]{article}

\usepackage{comment}

\usepackage[utf8]{inputenc}
\usepackage[english]{babel}
\usepackage[margin=1in]{geometry}

\usepackage{amsthm, amsfonts}
\usepackage{mathtools, thmtools}
\usepackage{bbm}
\usepackage{thm-restate}
\newtheorem{theorem}{Theorem}
\newtheorem{lemma}[theorem]{Lemma}

\theoremstyle{definition}

\usepackage[normalem]{ulem}
\usepackage{multirow}

\usepackage{algorithm}
\usepackage{algpseudocode}

\usepackage[numbers, sort]{natbib}

\usepackage[dvipsnames,usenames]{xcolor}
\usepackage[colorlinks=true,pdfpagemode=UseNone,urlcolor=RoyalBlue,linkcolor=RoyalBlue,citecolor=OliveGreen,pdfstartview=FitH]{hyperref}
\usepackage{prettyref}

\usepackage{graphicx}
\usepackage{subcaption}

\usepackage{tikz}
\usetikzlibrary{shapes, patterns, decorations, fit, intersections}

\usepackage{multicol}

\newcommand{\eqdef}{\overset{\mathrm{def}}{=\mathrel{\mkern-3mu}=}}

\title{Improved Bounds for Fractional Online Matching Problems}

\author{
	Zhihao Gavin Tang\thanks{ITCS, Shanghai University of Finance and Economics. Email: \texttt{tang.zhihao@mail.shufe.edu.cn}}
	\and
    Yuhao Zhang\thanks{Shanghai Jiao Tong University. Email: \texttt{zhang\_yuhao@sjtu.edu.cn}}
}
\date{}

\begin{document}

\maketitle

\begin{abstract}
Online bipartite matching with one-sided arrival and its variants have been extensively studied since the seminal work of Karp, Vazirani, and Vazirani (STOC 1990). 
Motivated by real-life applications with dynamic market structures, e.g. ride-sharing, two generalizations of the classical one-sided arrival model are proposed to allow non-bipartite graphs and to allow all vertices to arrive online. Namely, online matching with general vertex arrival is introduced by Wang and Wong (ICALP 2015), and fully online matching is introduced by Huang et al. (JACM 2020).

In this paper, we study the fractional versions of the two models.  
We improve three out of the four state-of-the-art upper and lower bounds of the two models. For fully online matching, we design a $0.6$-competitive algorithm and prove no algorithm can be $0.613$-competitive. For online matching with general vertex arrival, we prove no algorithm can be $0.584$-competitive.
Moreover, we give an arguably more intuitive algorithm for the general vertex arrival model, compared to the algorithm of Wang and Wong, while attaining the same competitive ratio of $0.526$. 
\end{abstract}

\section{Introduction}
\label{sec:intro}
Matching theory is a focal topic in the area of combinatorial optimization, with a wide range of applications. 
Online matching problems, arguably one of the most intriguing branches of matching theory, stays as a central point in the online algorithm literatures since the celebrated work of Karp, Vazirani, and Vazirani~\cite{stoc/KarpVV90}.
Karp et al. introduced the online bipartite matching problem with one-sided vertex arrival, i.e. one side of the vertices of an underlying bipartite graph arrive online and each vertex reveals its neighbors on its arrival. The algorithm makes immediate and irrevocable matching decisions for every online vertex and the goal is to maximize the size of the matching. It is shown that the na\"ive greedy algorithm achieves the optimal competitive ratio of $0.5$ among deterministic algorithms, and the Ranking algorithm achieves the optimal competitive ratio of $1-1/e$ among random algorithms. 
Since then, a fruitful line of research has been established based on the one-sided arrival model, including simplifications of the analysis of Ranking~\cite{soda/GoelM08,sigact/BenjaminM08,soda/DevanurJK13}, the vertex-weighted version~\cite{soda/AggarwalGKM11,talg/HuangTWZ19}, the edge-weighted version~\cite{teco/DevanurHKMY16, fahrbach2020edge, wine/FeldmanKMMP09, wine/KorulaMZ13,gao2021improved,blanc2021multiway}, the random arrival version~\cite{stoc/MahdianY11,stoc/KarandeMT11}, $b$-matching~\cite{tcs/KalyanasundaramP00}, and Adwords~\cite{jacm/MehtaSVV07, esa/BuchbinderJN07, stoc/DevanurJ12, focs/HuangZZ20}.

On the other hand, the one-sided arrival model is restricted and is not able to capture more dynamic scenarios in reality. To this end, two generalizations of the classical one-sided arrival model have been proposed that allow all vertices of a graph arriving online and allow non-bipartite graphs.


\paragraph{General Vertex Arrival.}
Wang and Wong~\cite{icalp/WangW15} proposed the online matching with general vertex arrival.
In this model, all vertices of an underlying graph arrive online. Upon the arrival of a vertex, the edges between it and all previously arrived vertices are revealed. The algorithm can either match the new vertex to its unmatched neighbors, or leave it unmatched for future.

The one-sided arrival model is a special case when offline vertices all arrive before online vertices.

\paragraph{Fully Online Matching.}
Huang et al.~\cite{jacm/HuangKTWZZ20} proposed the fully online matching problem. In this model, each time step can either be the arrival or the deadline of a vertex. On the arrival of a vertex, all edges between it and all previously arrived vertices are revealed. However, the algorithm does not need to make matching decisions at this point. On the deadline of a vertex, the algorithm must decide whether to match it to an unmatched neighbor, or leave it unmatched. It is assumed that all neighbors of a vertex arrives before its deadline.  

The one-sided arrival model is a special case when offline vertices arrive at the beginning and have deadlines at the end, and each online vertex has its deadline right after its arrival. Moreover, observe that any algorithm for the general vertex arrival model can be directly applied to the fully online matching model and would attain the same competitive ratio. 


\paragraph{Previous Results.} For the general vertex arrival model, Wang and Wong proposed a dual-based $0.526$-competitive fractional algorithm and proved an upper bound of $0.625$ for hardness. Gamlath et al.~\cite{focs/GamlathKMSW19} gave the first non-trivial randomized integral algorithm that achieves a competitive ratio of $0.5+\epsilon$. Buchbinder et al.~\cite{algorithmica/BuchbinderST19} improved the upper bound to $0.591$. 

For the fully online matching model, Huang et al. \cite{jacm/HuangKTWZZ20} showed that the Ranking algorithm is $0.5671$-competitive on bipartite graphs and $0.5211$-competitive on general graphs. Huang et al.~\cite{soda/HuangPTTWZ19} proved that the Water-filling algorithm is $2-\sqrt{2} \approx 0.585$-competitive for the fractional version of the problem. Recently, Huang et al.~\cite{focs/HuangTWZ20} beat the two classical algorithm Ranking and Water-filling by designing a $0.569$-competitive integral algorithm on bipartite graphs, and a $0.592$-competitive fractional algorithm on general graphs. On the negative side, Eckl et al.~\cite{orl/EcklKLS21} proved that no fractional algorithm can be better than $0.629$-competitive.


\subsection{Our Contributions}



Despite the growing interests on the general vertex arrival model and the fully online model, there are large gaps between the state-of-the-art upper and lower bounds. In this work, we advance the study of this line of research by focusing on the \emph{fractional} versions of the two settings.

On the positive side, we provide a $0.6$-competitive fractional algorithm for the fully online model, improving the previous $0.592$ competitive ratio of Huang et al.~\cite{focs/HuangTWZ20}, and a new $0.526$-competitive fractional algorithm for the general vertex arrival model, matching the competitive ratio of Wang and Wong~\cite{icalp/WangW15}.
On the negative side, We improve the hardness of the fully online model to $0.6132$ and the hardness of the general vertex arrival model to $0.584$, respectively. Both of them hold even for bipartite graphs. Refer to Figure~\ref{fig:results} for a summary of previous and our results. 
\begin{figure}[H]
    \center
    \includegraphics[scale = 0.55]{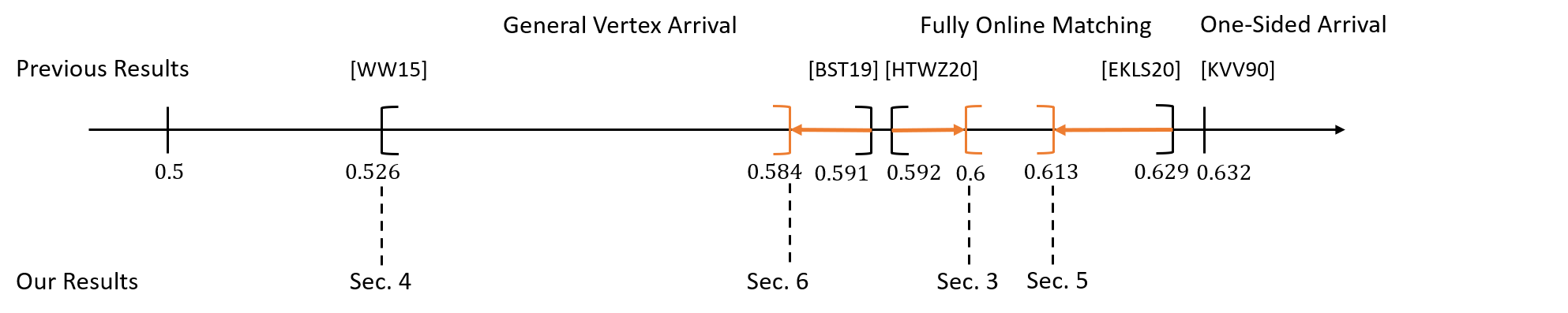}
    \caption{A summary of previous results and our results for the fractional online matching problems.}
    \label{fig:results}
\end{figure}

\paragraph{Economic View: History-based Pricing.}




Our algorithms are designed based on the economic view of the primal dual analysis. 
We start with the classical water-filling algorithm in the one-sided arrival model as an example. On the arrival of each online vertex, the algorithm continuously matches it to the least matched neighbor. Associated with the primal dual analysis, we interpret this algorithm as the following two steps: 1) each offline vertex $v$ maintains a dynamic price $f(x_v)$ that only depends on its current matched portion; 2) each online vertex $u$ continuously matches to the neighbor $v$ with the cheapest price and attains utility $1-f(x_v)$.
Based on it, Huang et al.~\cite{focs/HuangTWZ20} designed Eager Water-filling algorithm for the fully online model. We also present it in an economic view: 1) every vertex matches a dynamic price depending on its matched portion; 2) on the arrival of each vertex $u$, it continuously matches to the cheapest neighbor $v$ until $f(x_u)+f(x_v) > 1$; 3) Finally, on the deadline of each vertex, it continuously matches to the cheapest neighbor.

Unfortunately, prior to our work, there is no such an intuitive interpretation for the general vertex arrival model, although the $0.526$-competitive algorithm by Wang and Wong~\cite{icalp/WangW15} is also based on a primal dual analysis. A straightforward adaption of the eager water-filling algorithm to the general vertex arrival model is to apply the same first two steps, since there is no deadline in the general vertex arrival model.
However, this natural algorithm fails to achieve any non-trivial competitive ratio for arbitrary choice of $f$. (See Appendix~\ref{app:fail} for a detailed discussion.)

We resolve this issue and restore an intuitive economic interpretation for the general vertex arrival model by modifying the pricing step of the above algorithm. We introduce a novel pricing scheme called \emph{history-based pricing} that captures more information of the current matching status, than the current matched portion alone. Specifically, we use an extra number $a_u$ to record the active matched portion of each vertex $u$, i.e. the matched portion on its arrival, and then set a dynamic pricing of $g(a_u,x_u)$. The original pricing scheme is a special case of ours when the function $g$ is irrelevant from the first dimension. We recover the $0.526$-competitive fractional algorithm using the history-based pricing.

Furthermore, the same algorithmic idea leads to an immediate improvement of eager water-filling in the fully online model. By optimizing the function $g$ using factor-revealing LP techniques, we improve the competitive ratio from $0.592$ to $0.6$.

\paragraph{Hardness Results.}
For the fully online model, the previous hardness results of Huang et al.~\cite{jacm/HuangKTWZZ20} and Eckl et al.~\cite{orl/EcklKLS21} is only slightly smaller than $1-1/e$, and their constructions share the same structure of a tree plus an upper triangle at the end. Moreover, as observed by Eckl et al.~\cite{orl/EcklKLS21}, a similar structure cannot lead to a much smaller upper bound. 
We construct a different family of instances, that has a tree structure and an upper triangle structure alternatively. Our construction substantially improves the upper bound to $0.613$. We remark that Huang et al.~\cite{jacm/HuangKTWZZ20} used a similar alternating structure to show a tight upper bound of the Ranking algorithm. 

For the general vertex arrival mode, we build on the bad instance of \cite{icalp/WangW15}. Their construction is a two-stage instance that the adversary can choose when to stop. The first stage is a simple complete bipartite graph for which any $\Gamma$-competitive algorithm must match on average $\Gamma$ fraction for all vertices. Then they introduce extra vertices together with an upper-triangle structure as the second stage. The intuition is that all matchings made at the first stage are wrong and are punished by the second stage. However, since the algorithm knows the instance would end afterwards, it would greedily match vertices at the second stage.

Our idea is to further punish/forbidden the greedy behavior at the second stage by having a third stage. Though the intuition is straightforward, it is hard to implement since the matching status after the second stage need quite a few parameters to fully capture. We believe a tight hardness of the general arrival model would necessarily need multiple stages, and our construction would be useful for future studies of the problem.
Finally, we remark the two upper bounds are proved for any fractional algorithms, that directly imply upper bounds for any randomized algorithms.

\subsection{Other Related Works}
Simultaneous to the work of Huang et al.~\cite{jacm/HuangKTWZZ20}, Ashlagi et al.~\cite{ec/AshlagiBDJSS19} proposed the online windowed matching problem that is an edge-weighted version of fully online matching with the first-in-first-out assumption.
The general vertex arrival model has also been applied to the prophet matching setting and the secretary matching setting by Ezra et al.~\cite{ec/EzraFGT20, corr/EzraFGT20}. In the prophet setting, each edge has a weight that is drawn independently from a priori known distribution. In the secretary setting, edges are also weighted and all vertices arrive in a random order. The goal is to maximize the total weight of the matching in both models. Ezra et al. give an optimal $\frac{1}{2}$-competitive algorithm for the prophet setting and an optimal $\frac{5}{12}$-competitive algorithm for the secretary setting.

Besides two arrival models studied in this paper, researches have also studied the edge arrival model. Gamlath et al.~\cite{focs/GamlathKMSW19} gave a strong negative result showing that the $0.5$-competitive greedy algorithm is the best possible. Buchbinder, Segev, and Tkach~\cite{algorithmica/BuchbinderST19} provided a positive result in the edge arrival setting when the underlying graph is a forest. Epstein et al.~\cite{stacs/EpsteinLSW13} studied a relaxed setting of edge arrival by allowing free disposals of edges and designed constant competitive algorithms for weighted graphs.

\section{Preliminaries}
In this section, we present the two models formally and then introduce the standard primal dual analysis for matching.
\paragraph{General Vertex Arrival.}
Consider an underlying input graph $G=(V,E)$, where $V$ arrives online and we will reveal $E$ along with $V$. At each $v$'s arrival time:
\begin{itemize}
    \item All the edges between $v$ and previously arrived vertices are revealed. 
    \item The algorithm fractionally matches the newly revealed edges. And these edges cannot be matched afterwards. 
\end{itemize}
Every vertex cannot be matched more than one unit. The goal is to maximize the size of the fractional matching. 
\paragraph{Fully Online Matching.}
Consider an underlying input graph $G=(V,E)$, where $V$ arrives online and we will reveal $E$ along with $V$. Each time step is one of the following:
\begin{itemize}
    \item The arrival of vertex $v$: all edges between $v$ and previously arrived vertices are revealed. 
    \item The deadline of vertex $v$: this is the last chance we can (fractionally) match $v$ to the remaining vertices in the graph.
\end{itemize}
We assume that all neighbors of $v$ arrive before $v$'s deadline. Every vertex cannot be matched more than one unit. The goal is to maximize the size of the fractional matching. 
As observe by Huang et al.~\cite{jacm/HuangKTWZZ20}, it is without loss of generality to consider lazy algorithms that only make matching decisions on the deadlines of vertices. On the other hand, we study algorithms that make matching decisions on the arrival of vertices for a cleaner analysis.

\paragraph{Primal Dual Analysis} Following the standard competitive analysis, we say an online algorithm is $\Gamma$-competitive if for any input instance, its solution is always at least $\Gamma$ times the optimal solution. We apply the standard primal dual framework. First, we present the standard matching LP and its dual, that is common for the two models we study in this paper.
\begin{align*}
\max: \quad & \textstyle \sum_{(u,v)\in E} x_{uv} && \qquad\qquad & \min: \quad & \textstyle\sum_{u \in V} \alpha_u\\
\text{s.t.} \quad & \textstyle \sum_{v:(u,v)\in E} x_{uv} \leq 1 && \forall u\in V & \text{s.t.} \quad & \alpha_u + \alpha_v \geq 1 && \forall (u,v)\in E \\
& x_{uv} \geq 0 && \forall (u,v)\in E & & \alpha_u \geq 0 && \forall u \in V
\end{align*}

We shall set the primal variables $x_{uv}$'s to be the matched portion of edge $(u,v)$ of our algorithm, and maintain the dual variables at the same time.
\begin{lemma}
    \label{lem:primal-dual}
    An online primal dual algorithm is $\Gamma$-competitive if we have:
    \begin{itemize}
        \item \emph{Approximate dual feasibility:~}
            $\forall (u, v) \in E,~\alpha_u + \alpha_v \ge \Gamma$;
        \item \emph{Reverse weak duality:~} the primal objective $P$ and the dual objective $D$ satisfies that $P \ge D $.
    \end{itemize}
\end{lemma}

\section{Improved Algorithm for Fully Online Matching}
\label{sec:fully_lower}

We first recall the Eager Water-filling algorithm proposed by Huang et al.~\cite{focs/HuangTWZ20} and present it from an economic view.
\begin{itemize}
	\item Each vertex $v$ maintains a price that depends on its current water level.
	\item On the arrival of each vertex $u$, fractionally match it to the neighbor with the cheapest price, until the cheapest price among all neighbors and $u$'s own price sum to $1$.
	\item On the departure of each vertex $u$, fractionally match it to the neighbor with the cheapest price to maximize its own utility.
\end{itemize}

The standard water-filling algorithm only consists of the first and third steps. In the design of Water-filling~\cite{soda/HuangPTTWZ19} and Eager Water-filling~\cite{focs/HuangTWZ20}, the price of a vertex solely depends on its water level, i.e. the price of a vertex $v$ = $f(x_v)$ for some increasing function $f$. As a result, the selection of a vertex with the cheapest price is equivalent to the selection of a vertex with the lowest water level.

We improve upon the Eager Water-filling algorithm by modifying its pricing policy.  
In principle, one can set a price based on all historical information of the instance, instead of the specific local information of its water level. We implement this idea by introducing a two-dimensional price function. Our design of the algorithm is directly motivated by the primal-dual analysis.

\paragraph{Our Algorithm.} Fix an increasing function $f:[0,1] \to [0,1]$ and a two-dimensional function $g:[0,1]^2 \to [0,1]$ that is non-decreasing in both dimensions. 
All matched portion $x_{uv}$'s, $x_u$'s, and the dual variable $\alpha_u$'s  are initialized to $0$ upon $u$'s arrival. 
\begin{itemize}
	\item On the arrival of each vertex $u$, fractionally matches $u$ to the neighbor $v$ with the cheapest price $\rho_v = g(a_v, x_v)$, as long as $f(x_u) + g(a_v, x_v) \le 1$. When the process ends, let $a_u = x_u$ be the \emph{active} water level of $u$.
	\item On the departure of each vertex $u$, fractionally match it to the neighbor $v$ with the cheapest price $g(a_v, x_v)$.
\end{itemize}

In both steps, when $x_{uv}$ increases by $dx$ at the arrival or departure of $u$, update $\alpha_u, \alpha_v$ by 
\[
d\alpha_u = 1-g(a_v,x_v) dx \quad \text{and} \quad d\alpha_v = g(a_v, x_v) dx
\]
Refer to Algorithm~\ref{alg:fully} for the pseudocode of our algorithm.
\paragraph{Remark.} Eager Water-filling algorithm is a special case of the above algorithm when $g$ is set to be the same function $f$, i.e. $g(a, x) = f(x)$ for all $a$.

\begin{algorithm}
	\caption{History-based Pricing Algorithm for Fully Online Matching}
	\label{alg:fully}
	\begin{algorithmic}
	\If{a vertex $u$ arrives} 
		\For{$\forall v \in N(u)$} \Comment{$N(u)$ denotes $u$'s current neighbor set}
		\State $x_{u,v} \gets 0$ \Comment{$x_{u,v}$ denotes the matched portion of $(u,v)$}
		\EndFor
		\State $x_{u} \gets 0$ \Comment{$x_u$ denotes the water level of $u$}
		\State $v \gets \arg\min_{v\in N(u)}g(a_v,x_{v})$
		\While{$f(x_u) + g(a_v,x_v)\leq 1$}
			\State $x_{u,v}  ~{\scriptstyle  \mathrel{+}=}~ dx$ \Comment{Match $u$, $v$ (small enough portion $dx$).}
			\State $x_{u} ~{\scriptstyle  \mathrel{+}=}~ dx$, $x_{v}  ~{\scriptstyle  \mathrel{+}=}~ dx$ \Comment{Water level update}
			\State $\alpha_u   ~{\scriptstyle  \mathrel{+}=}~ (1 - g(a_{v},x_{v}))dx$,  $\alpha_{v}  ~{\scriptstyle  \mathrel{+}=}~ g(a_{v},x_{v})dx$ \Comment{
				Dual update}
			\State $v \gets \arg\min_{v\in N(u)}g(a_v,x_{v})$.
		\EndWhile
		\State $a_u \gets x_u$ \Comment{Fix $u$'s active water level.}
	\EndIf
	\If{a vertex $u$ departs} 
		\While{$x_u < 1$}
			\State $v \gets \arg\min_{v\in N(u)}g(a_v,x_{v})$ \Comment{$N(u)$ becomes $u$'s final neighbor set}
			\State $x_{u,v}  ~{\scriptstyle  \mathrel{+}=}~ dx$ \Comment{Match $u$, $v$ (small enough portion $dx$).}
			\State $x_{u} ~{\scriptstyle  \mathrel{+}=}~ dx$, $x_{v}  ~{\scriptstyle  \mathrel{+}=}~ dx$ \Comment{Water level update}
			\State $\alpha_u   ~{\scriptstyle  \mathrel{+}=}~ (1 - g(a_{v},x_{v}))dx$,  $\alpha_{v}  ~{\scriptstyle  \mathrel{+}=}~ g(a_{v},x_{v})dx$ \Comment{
				Dual update}
		\EndWhile
	\EndIf
	\end{algorithmic}
\end{algorithm}

\paragraph{Analysis.} 
Our proofs are inherently similar to the analysis of Eager Water-filling by \cite{focs/HuangTWZ20}. Each lemma below has a corresponding counterpart in the analysis of Eager Water-filling. After all, the family of algorithms we proposed include Eager Water-filling as a special case. Nevertheless, we need to do quite non-trivial transformation of the derived lower bounds, so that we can apply factor revealing LP techniques to optimize the functions $f,g$ of our algorithm. 

Fix a pair of neighbors $(u,v)$. Suppose $u$'s deadline is before $v$'s deadline. Let $p_u$ be the water level of $u$ right \emph{before} $u$'s deadline. 
Let $p_v$ be the water level of $v$ right \emph{after} $v$'s deadline. 

\begin{lemma}
Right after $u$'s deadline, we have
\begin{equation}
\label{eq:gain}
\alpha_u + \alpha_v \ge a_u \cdot f(a_u) + \int_{a_u}^{p_u} g(a_u, x_u) dx_u + a_v \cdot f(a_v) + \int_{a_v}^{p_v} g(a_v, x_v) dx_v + (1-p_u) \cdot (1-g(a_v, p_v)).
\end{equation}
\end{lemma}

\begin{proof}
On the arrival of the vertex $v$, it fractionally matches the neighbor $z$ with the cheapest price as long as $f(x_v) + g(a_z, x_z) \le 1$. Hence, during the process, the marginal gain of $v$ is $d \alpha_v = (1-g(a_z, x_z)) dx \ge f(a_v) dx$. Here, the inequality holds since the utility $1-g(a_z,x_z)$ decreases during the process as the cheapest price increases. We use the fact that $f$ is a continuous function. 

Thus, $\alpha_v \ge a_v \cdot f(a_v)$ after $v$'s arrival. Later, $v$ is passively matched and $x_v$ increases from $a_v$ to $p_v$, $\alpha_u$ increases at the rate of $g(a_v, x_v)$. Hence, $\alpha_v \ge a_v \cdot f(a_v) + \int_{a_v}^{p_v}g(a_v, x_v)dx_v$ right after $u$'s deadline.

Similarly, $\alpha_u \ge a_u \cdot f(a_u) + \int_{a_u}^{p_u}g(a_u, x_u)dx_u$ right \emph{before} $u$'s deadline.
On the deadline of $u$, $u$ actively matches the neighbor with the cheapest price, which is at most $g(a_v, p_v)$. Therefore, after the deadline of $u$, $\alpha_u \ge a_u \cdot f(a_u) + \int_{a_u}^{p_u}g(a_u, x_u)dx_u + (1-p_u) \cdot (1-g(a_v, p_v))$.

Summing up the total gains of $\alpha_u, \alpha_v$ gives the claimed bound.
\end{proof}

\paragraph{Reformulating the lower bound.}
To facilitate the analysis, we define the following function $h: [0,1]^2 \to [0,1]$ that has a one-one correspondence with $g$:
\[
h(\tau,\theta) = \inf \{x: g(f^{-1}(\tau), x) \ge \theta\}. 
\]
Our function $g$ shall be only defined on $(a, x) \in [0,1]^2$ with $x \ge a$. Hence, $h(\tau,\tau)=f^{-1}(\tau)$. Moreover, our functions $f,g$ satisfy that $f(x) = g(x,x)$ for all $x \in [0,1]$.

Then for any $0 \le \tau \le \theta \le 1$, we have
\[
\int_{h(\tau,\tau)}^{h(\tau,\theta)} g(f^{-1}(\tau),x) dx = \theta \cdot h(\tau, \theta) - \tau \cdot h(\tau, \tau) - \int_{\tau}^{\theta} h(\tau, y) dy.
\]

The following lemma studies the two cases depending on whether $u$ arrives earlier than $v$ or $v$ arrives earlier than $u$ and reformulate the lower bound~\eqref{eq:gain}.

\begin{lemma}
\label{lem:uearlier}
If $u$ arrives earlier than $v$, then
\begin{multline}
\label{eq:uearlier}
\alpha_u + \alpha_v \ge \min_{\substack{\tau_u \le \theta_u \\ 1-\theta_u \le \tau_v \le \theta_v}} \bigg\{ \theta_u \cdot h(\tau_u, \theta_u) - \int_{\tau_u}^{\theta_u} h(\tau_u, y_u) dy_u \\
+ \theta_v \cdot h(\tau_v, \theta_v) - \int_{\tau_v}^{\theta_v} h(\tau_v, y_v) dy_v + (1-h(\tau_u, \theta_u)) \cdot (1-\theta_v) \bigg\}
\end{multline}
\end{lemma}
\begin{proof}
Let $\tau_u = f(a_u), \tau_v= f(a_v), \theta_u = g(a_u, p_u)$ and $\theta_v = g(a_v, p_u)$. 
With the notion of $h$, we rewrite Equation~\eqref{eq:gain}.
\begin{multline}
\label{eq:uearlier_gainv}
a_v \cdot f(a_v) + \int_{a_v}^{p_v} g(a_v, x_v) dx_v = \tau_v \cdot h(\tau_v,\tau_v) + \int_{h(\tau_v,\tau_v)}^{p_v} g(f^{-1}(\tau_v), x_v) dx_v \\
\ge \tau_v \cdot h(\tau_v,\tau_v) + \int_{h(\tau_v,\tau_v)}^{h(\tau_v, \theta_v)} g(f^{-1}(\tau_v), x_v) dx_v = \rho_v \cdot h(\alpha_v, \rho_v) - \int_{\alpha_v}^{\rho_v} h(\alpha_v,y_v) dy_v,
\end{multline}
where the inequality holds from the fact that $p_v \ge h(\tau_v,\theta_v)$.
Let $G$ be the right hand side of Equation~\eqref{eq:gain}. We calculate the derivative of $G$ over $p_u$:
\[
\frac{\partial G}{\partial p_u} = g(a_u, p_u) - 1 + g(a_v, p_v) \ge g(a_u, p_u) - 1 + f(a_v) \ge 0.
\]
The last inequality holds due to the stopping condition of our algorithm on $v$'s arrival. Consequently, the minimum of $G$ is achieved when $p_u = h(\tau_u,\theta_u)$. 
Thus, we have 
\begin{align}
&\phantom{\ge} a_u \cdot f(a_u) + \int_{a_u}^{p_u} g(a_u, x_u) dx_u + (1-p_u) \cdot (1-g(a_v,p_v)) \notag \\
& \ge \tau_u \cdot h(\tau_u,\tau_u) + \int_{h(\tau_u,\tau_u)}^{h(\tau_u,\theta_u)} g(f^{-1}(\tau_u),x_u) dx_u + (1-h(\tau_u,\theta_u)) \cdot (1-\theta_v) \notag \\
& = \theta_u \cdot h(\tau_u, \theta_u) - \int_{\tau_u}^{\theta_u} h(\tau_u, y_u) dy_u + (1-h(\tau_u, \theta_u)) \cdot (1-\theta_v). \label{eq:uearlier_gainu}
\end{align}
We conclude the proof by summing up \eqref{eq:uearlier_gainv} and \eqref{eq:uearlier_gainu} and taking the minimum over all possible values of $\tau_u,\theta_u,\tau_v$ and $\theta_v$. Notice that since $u$ arrives earlier than $v$, we must have $\tau_v = f(a_v) > 1-g(a_u,p_u) = 1-\theta_u$, by the design of our algorithm.
\end{proof}

\begin{lemma}
\label{lem:vearlier}
If $v$ arrives earlier than $u$, then
\begin{equation}
\label{eq:vearlier}
	\alpha_u + \alpha_v \ge \min_{\substack{\tau_v \le \theta_v}} \left\{ \theta_v \cdot h(\tau_v, \theta_v) - \int_{\tau_v}^{\theta_v} h(\tau_v, y_v) dy_v + 1 - \theta_v \right\}
\end{equation}
\end{lemma}

\begin{proof}
Let $\tau_u = f(a_u), \tau_v = f(a_v), \theta_u = g(a_u,p_u)$ and $\theta_v = g(a_v,p_u)$. Let $G$ be the right hand side of Equation~\eqref{eq:gain}. We calculate the derivative of $G$ over $p_u$:
\[
\frac{\partial G}{\partial p_u} = g(a_u, p_u) - 1 + g(a_v, p_v) \ge f(a_u) - 1 + g(a_v,p_u) \ge 0.
\]
The last inequality holds due to the stopping condition of our algorithm on $u$'s arrival. Thus, the minimum is achieved when $p_u=a_u$. Thus, we have
\begin{align*}
\alpha_u+\alpha_v & \ge a_v \cdot f(a_v) + \int_{a_v}^{p_v} g(a_v, x_v) dx_v + a_u \cdot f(a_u) + (1-a_u) \cdot (1-g(a_v, p_v)) \\
& \ge \tau_v \cdot h(\tau_v, \tau_v) + \int_{h(\tau_v,\tau_v)}^{h(\tau_v, \theta_v)} g(f^{-1}(\tau_v), x_v) dx_v +(1-\theta_v) \\
& = \theta_v \cdot h(\tau_v, \theta_v) - \int_{\tau_v}^{\theta_v} h(\tau_v, y_v) dy_v +1-\theta_v,
\end{align*}
where the second inequality holds by the fact that $f(a_u) \ge 1-g(a_v,p_v) = 1-\theta_v$.
Finally, taking the minimum over all possible values of $\tau_v,\theta_v$ concludes the proof.
\end{proof}

Let $\Phi_1, \Phi_2$ denote the two lower bounds derived from the above lemmas. Formally,
\begin{align*}
    &\Phi_1(\tau_v,\theta_v) \eqdef \theta_v \cdot h(\tau_v, \theta_v) - \int_{\tau_v}^{\theta_v} h(\tau_v, y_v) dy_v + 1-\theta_v; \\ 
    &\Phi_2(\tau_u,\theta_u,\tau_v,\theta_v) \eqdef \left(
    \begin{aligned}
        &\theta_u \cdot h(\tau_u, \theta_u) - \int_{\tau_u}^{\theta_u} h(\tau_u, y_u) dy_u + \theta_v \cdot h(\tau_v, \theta_v)\\
        & - \int_{\tau_v}^{\theta_v} h(\tau_v, y_v) dy_v + (1-h(\tau_u, \theta_u)) \cdot (1-\theta_v)
    \end{aligned}
    \right).
\end{align*}
Finally, we use the folklore factor revealing lp techniques to optimize the function $h$. A detailed implementation of the factor revealing lp is provided in Appendix~\ref{app:factor_lp}.
\begin{lemma}
\label{lem:opth}
There exists a function $h:[0,1]^2 \to [0,1]$ such that for $\Gamma = 0.6$:
\begin{align}
	& \forall 0\le \tau_v \le \theta_v \le 1: &\Phi_1 \ge \Gamma; \label{eqn:h-feasible-1} \\
    & \forall 0\le \tau_u \le \theta_u \le 1, 1-\theta_u\le \tau_v \le \theta_v \le 1 :\quad  & \Phi_2 \ge \Gamma; \label{eqn:h-feasible-2} 
\end{align}
\end{lemma}

\begin{theorem}
	The algorithm with the functions $f,g$ chosen corresponding to the $h$ function of Lemma~\ref{lem:opth} is $\Gamma=0.6$-competitive for fractional fully online matching on general graphs.
\end{theorem}
\begin{proof}
We conclude the competitive ratio of our algorithm by putting the lemmas together. Fix an arbitrary pair of neighbors $(u,v)$ with $u$'s deadline later than $v$'s deadline. Approximate dual feasibility follows by the two cases. If $u$ arrives earlier than $v$,
\[
\alpha_u + \alpha_v \ge \min_{\substack{0\le\tau_u \le \theta_u\le 1 \\ 1-\theta_u \le \tau_v \le \theta_v}} \Phi_2(\tau_u,\theta_u,\tau_v,\theta_v) \ge 0.6,
\]
where the first inequality follows from Lemma~\ref{lem:uearlier} and the second inequality follows from Lemma~\ref{lem:opth}.
If $v$ arrives earlier than $u$,
\[
\alpha_u + \alpha_v \ge \min_{0\le\tau_v\le\theta_v\le 1} \Phi_1(\tau_v,\theta_v) \ge 0.6,
\]
where the first inequality follows from Lemma~\ref{lem:vearlier} and the second inequality follows from Lemma~\ref{lem:opth}.
Finally, reverse weak duality holds trivially according to the definition of dual variables.
\end{proof}

\section{Fractional Algorithm for General Vertex Arrival}
\label{sec:general_lower}

In this section, we provide a primal-based algorithm for the fractional online matching problem with general vertex arrival by Wang and Wong~\cite{icalp/WangW15}. Our algorithm has the same worst case competitive ratio $0.526$ as Wang and Wong~\cite{icalp/WangW15}. 

We argue that our algorithm is more intuitive and we believe it would be helpful for designing integral algorithms for the problem. Our algorithm is inspired by the algorithm we proposed for the fully online matching problem. Indeed, it is the same algorithm but limited to the matching decisions made on the arrivals of vertices.
\paragraph{Our Algorithm.}
Fix an increasing function $f:[0,1] \to [0,1]$ and a two-dimensional function $g:[0,1] \to [0,1]$.
\begin{itemize}
	\item On the arrival of each vertex $u$, fractionally matches $u$ to the neighbor with the cheapest price $g(a_v, x_v)$, as long as $f(x_u) + g(a_v,x_v) \le 1$. When the process ends, let $a_u = x_u$ be the \emph{active} water level of the vertex $u$.
\end{itemize}
When $x_{uv}$ is increased by $dx$ on the arrival of vertex $u$, update $\alpha_u, \alpha_v$ by
\[
d\alpha_u = (1-g(a_v, x_v)) dx \quad \text{and} \quad d\alpha_v =g(a_v, x_v) dx.
\]

Refer to Algorithm~\ref{alg:general} for the pseudocode of our algorithm.

\begin{algorithm}[H]
	\caption{History-based Pricing Algorithm for General Vertex Arrival}\label{alg:general}
	\begin{algorithmic}
	\If{a vertex $u$ arrives} 
		\For{$\forall v \in N(u)$} \Comment{$N(u)$ denotes $u$'s current neighbor set}
		\State $x_{u,v} \gets 0$ \Comment{$x_{u,v}$ denotes the matched portion of $(u,v)$}
		\EndFor
		\State $x_{u} \gets 0$ \Comment{$x_u$ denotes the water level of $u$}
		\State $v \gets \arg\min_{v\in N(u)}g(a_v,x_{v})$
		\While{$f(x_u) + g(a_v,x_v)\leq 1$}
			\State $x_{u,v}  ~{\scriptstyle  \mathrel{+}=}~ dx$ \Comment{Match $u$, $v$ (small enough portion $dx$).}
			\State $x_{u} ~{\scriptstyle  \mathrel{+}=}~ dx$, $x_{v}  ~{\scriptstyle  \mathrel{+}=}~ dx$ \Comment{Water level update}
			\State $\alpha_u   ~{\scriptstyle  \mathrel{+}=}~ (1 - g(a_{v},x_{v}))dx$,  $\alpha_{v}  ~{\scriptstyle  \mathrel{+}=}~ g(a_{v},x_{v})dx$ \Comment{
				Dual update}
			\State $v \gets \arg\min_{v\in N(u)}g(a_v,x_{v})$.
		\EndWhile
		\State $a_u \gets x_u$ \Comment{Fix $u$'s active water level.}
	\EndIf
	\end{algorithmic}
\end{algorithm}

\paragraph{Analysis.} Fix a pair of neighbors $(u,v)$. Suppose $v$ arrives earlier than $u$. Let $p_v$ be the water level of $v$ right after $u$'s arrival. 
\begin{lemma}
\label{lem:gain_general_arrival}
Right after $u$'s arrival,
\[
\alpha_u + \alpha_v \ge a_v \cdot f(a_v) + \int_{a_v}^{p_v} g(a_v, x_v) dx_v + a_u \cdot f(a_u)
\]	
\end{lemma}
\begin{proof}
	On the arrival of vertex $v$, it fractionally matches to the neighbor with the cheapest price and $d\alpha_v = (1-g(a_z,x_z)) dx \ge f(a_v) dx$. Later, $v$ is passively matched and the water level of $v$ increases from $a_v$ to $p_v$, $\alpha_v$ increases at the rate of $g(a_v, x_v)$. Hence $\alpha_v \ge a_v \cdot f(a_v) + \int_{a_v}^{p_v} g(a_v,x_v)dx_v$.
	Similarly, the gain of $u$ is at least $a_u \cdot f(a_u)$. Putting the two lower bounds together concludes the lemma.
\end{proof}

\paragraph{Reformulating the bound.} Again, we introduce a function $h: [0,1]^2 \to [0,1]$ that has a one-one correspondence with $g$:
\[
h(\tau,\theta) = \inf \{x: g(f^{-1}(\tau), x) \ge \theta\}. 
\]
Our function $g$ shall be only defined on $(a, x) \in [0,1]^2$ with $x \ge a$. Hence, $h(\tau,\tau)=f^{-1}(\tau)$. Moreover, our functions $f,g$ satisfy that $f(x) = g(x,x)$ for all $x \in [0,1]$.

Then for any $0 \le \tau \le \theta \le 1$, we have
\[
\int_{h(\tau,\tau)}^{h(\tau,\theta)} g(f^{-1}(\tau),x) dx = \theta \cdot h(\tau, \theta) - \tau \cdot h(\tau, \tau) - \int_{\tau}^{\theta} h(\tau, y) dy.
\]

\begin{lemma}
\label{lem:gain_general}
Right after $u$'s arrival,
\[
\alpha_u + \alpha_v \ge \min_{\substack{\tau_v \le \theta_v \\ \tau_u \ge 1-\theta_v}} \left\{ \theta_v \cdot h(\tau_v, \theta_v) - \int_{\tau_v}^{\theta_v} h(\tau_v, y_v) dy_v + \tau_u \cdot h(\tau_u,\tau_u) \right\}
\]
\end{lemma}
\begin{proof}
Let $\tau_v = f(a_v), \tau_u = f(a_u)$ and $\theta_v = g(a_v, p_v)$. By Lemma~\ref{lem:gain_general_arrival}, we have 
\begin{multline*}
	\alpha_u + \alpha_v = \tau_v \cdot h(\tau_v, \tau_v) + \int_{h(\tau_v,\tau_v)}^{p_v} g(f^{-1}(\tau_v), x_v) dx_v + \tau_u \cdot h(\tau_u,\tau_u) \\
	\ge \tau_v \cdot h(\tau_v, \tau_v) + \int_{h(\tau_v,\tau_v)}^{h(\tau_v, \theta_v)} g(f^{-1}(\tau_v), x_v) dx_v + \tau_u \cdot h(\tau_u,\tau_u) \\
	= \theta_v \cdot h(\tau_v, \theta_v) - \int_{\tau_v}^{\tau_v} h(\tau_v, y_v) dy_v + \tau_u \cdot h(\tau_u, \tau_u),
\end{multline*}
where the inequality holds by the fact that $p_v \ge h(\tau_v,\theta_v)$, by the definition of $h$.
Notice that $\tau_u =f(a_u) > 1-g(a_v,p_v) = 1- \theta_v$ after the arrival of $u$, according to the design of our algorithm. We conclude the lemma by taking the minimum over all possible values of $\tau_v, \theta_v$ and $\tau_u$.
\end{proof}

We are left to optimize the function $h$ and this is again solved by factor revealing lp. We omit the detailed implementation. 
Our experiment suggests the best ratio to be $\Gamma\approx 0.526$ using our algorithm and analysis. Though we do not optimize the function $h$ in a closed form, we believe it is the same competitive ratio of Wang and Wong~\cite{icalp/WangW15}.
\begin{lemma}
\label{lem:opth_general_arrival}
There exists a function $h:[0,1]^2 \to [0,1]$ such that for $\Gamma = 0.526$:
\begin{align*}
	\forall 0\le \tau_v \le \theta_v \le 1:\quad  & \theta_v \cdot h(\tau_v, \theta_v) - \int_{\tau_v}^{\theta_v} h(\tau_v, y_v) dy_v + (1-\theta_v) \cdot h(1-\theta_v, 1-\theta_v) \ge \Gamma;\\
	& \tau \cdot h(\tau, \tau) \text{ is increasing.}
\end{align*} 
\end{lemma}

\begin{theorem}
The algorithm with the functions $f,g$ chosen corresponding to the $h$ function of Lemma~\ref{lem:opth_general_arrival} is $\Gamma = 0.526$-competitive for fractional online matching with general vertex arrival.
\end{theorem}
\begin{proof}
We conclude the competitive ratio of our algorithm by putting the lemmas together. Fix an arbitrary pair of neighbors $(u,v)$ with $u$ arrives later than $v$. By Lemma~\ref{lem:gain_general_arrival},
\begin{multline*}
\alpha_u + \alpha_v \ge \min_{\substack{\tau_v \le \theta_v \\ \tau_u \ge 1-\theta_v}} \left\{ \theta_v \cdot h(\tau_v, \theta_v) - \int_{\tau_v}^{\theta_v} h(\tau_v, y_v) dy_v + \tau_u \cdot h(\tau_u,\tau_u) \right\} \\
= \min_{\tau_v \le \theta_v} \left\{ \theta_v \cdot h(\tau_v, \theta_v) - \int_{\tau_v}^{\theta_v} h(\tau_v, y_v) dy_v + (1-\theta_v) \cdot h(1-\theta_v,1-\theta_v) \right\} \ge 0.526,
\end{multline*}
where the first inequality follows from Lemma~\ref{lem:gain_general_arrival}, the equality follows from the monotonicity of $h$ from Lemma~\ref{lem:opth_general_arrival}, and the last inequality is also from Lemma~\ref{lem:opth_general_arrival}.

Finally, reverse weak duality holds trivially according to the definition of dual variables.
\end{proof}

\section{Improved Hardness for Fully Online Matching}
\label{sec:fully_upper}

In this section, we show an improved upper bound for the fractional fully online matching problem. Previously, the best known upper bound is $0.6297$ by Eckl et al.~\cite{orl/EcklKLS21}. 

\begin{theorem}
No algorithm is strictly better than $0.6132$-competitive for the fractional fully online matching problem, even for bipartite graphs.
\end{theorem}

\paragraph{Hard Instance.} Our construction includes $2n \cdot \ell$ vertices, that are divided into $\ell$ identical groups. Each group consists of $4$ parts, that are denoted by $A_k, B_k, C_k$, and $D_k$. Each group of $A_k, B_k$ has $\alpha \cdot n$ vertices and $C_k, D_k$ has $(1-\alpha) \cdot n$ vertices, where $\alpha$ is a constant to be fixed later.

We first describe the structure of the underlying graph and then define the arrivals and departures of the vertices. Refer to Figure~\ref{fig:upper}.


\begin{itemize}
    \item Upper triangle between $A_k$ and $B_k$. Let $A_k = \{a_{ki}\}_{i}$ and $B_k = \{b_{ki}\}_{i}$. The $i$-th vertex $a_{ki}$ in $A_k$ is connected to every $b_{kj}$ for $j\geq i$.
    \item Complete bipartite graph between $A_k$ and $C_k$. 
    \item Complete bipartite graph between $C_k$ and $D_k \cup B_{k+1} \cup C_{k+1}$.
\end{itemize}
 
\begin{figure}[H]
    \center
    \includegraphics[scale = 0.5]{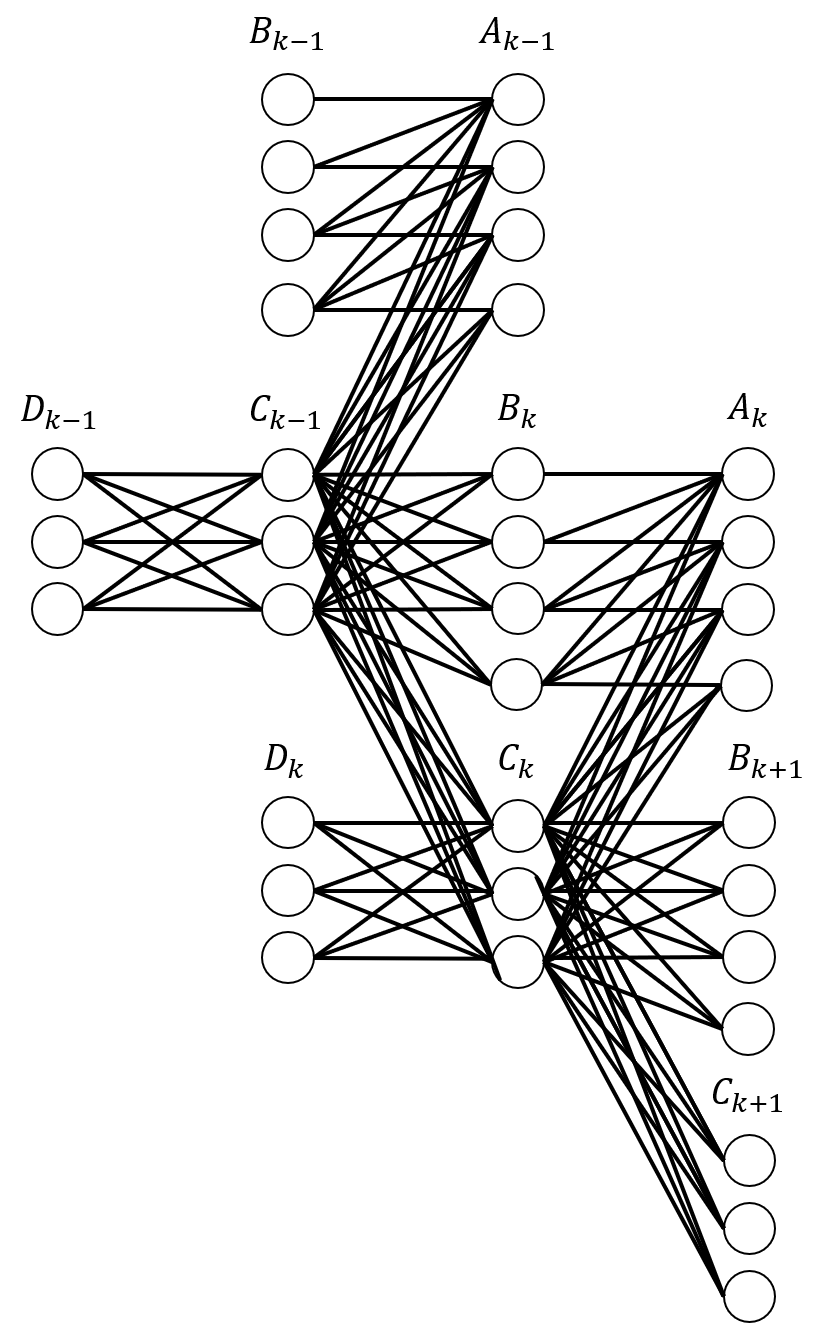}
    \caption{An illustration when $\alpha \cdot n =4$ and $(1-\alpha)\cdot n=3$.}
    \label{fig:upper}
\end{figure}

Next, we define the arrivals and deadlines of the vertices. Let there be $\ell$ stages. At the end of the $(k-1)$-th stage, all vertices in $\{A_i\}_{i \in [k-1]}, \{B_i\}_{i \in [k]}, \{C_i\}_{i \in [k]}, \{D_i\}_{i \in [k-1]}$ have arrived. Furthermore, all vertices in $\{A_i\}_{i \in [k-1]}, \{B_i\}_{i \in [k-1]}, \{C_i\}_{i \in [k-1]}, \{D_i\}_{i \in [k-1]}$ have left. In other words, the only remaining vertices in the graph belongs to $B_k \cup C_k$. Refer to Figure~\ref{fig:construction-1}. 

We shall define the arrivals and deadlines in the $k$-th stage.
\begin{enumerate}
    \item (Fig~\ref{fig:construction-2}) The vertices of $A_{k}$ arrive and leave immediately one by one. Upon the arrival and departure of $a_{ki}$, $a_{ki}$'s neighbors $\{b_{kj}\}_{j\ge i} \cup C_k$ are still alive (i.e., the vertex arrives but not depart) in the graph.
    \item (Fig~\ref{fig:construction-2}) The vertices of $B_k$ leave. There is no alive neighbor for vertices in $B_k$. 
    \item (Fig~\ref{fig:construction-3}) The vertices of $D_k \cup B_{k+1} \cup C_{k+1}$ arrive and the vertices of $C_k$ leave ony by one. $c_{ki}$'s neighbors $D_k \cup B_{k+1} \cup C_{k+1}$ are still alive.
    \item (Fig~\ref{fig:construction-3}) The vertices of $D_k$ leave. There is no alive neighbor for vertices in $D_k$.  
\end{enumerate}

\begin{figure}[H]
    \center
    \begin{subfigure}[H]{.15\textwidth}
        \includegraphics[scale = 0.5]{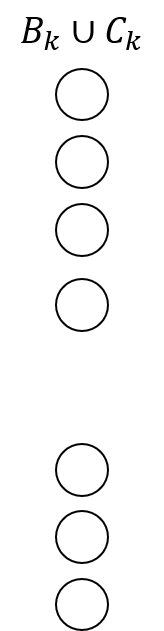}
        \caption{}
        \label{fig:construction-1}
    \end{subfigure}
    \begin{subfigure}[H]{.15\textwidth}
        \includegraphics[scale = 0.5]{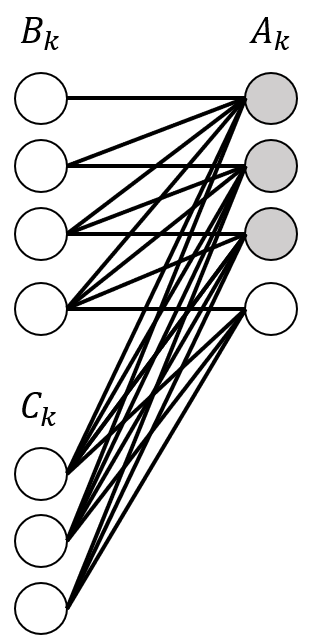}
        \caption{}
        \label{fig:construction-2}
    \end{subfigure}
    \begin{subfigure}[H]{.3\textwidth}
        \includegraphics[scale = 0.5]{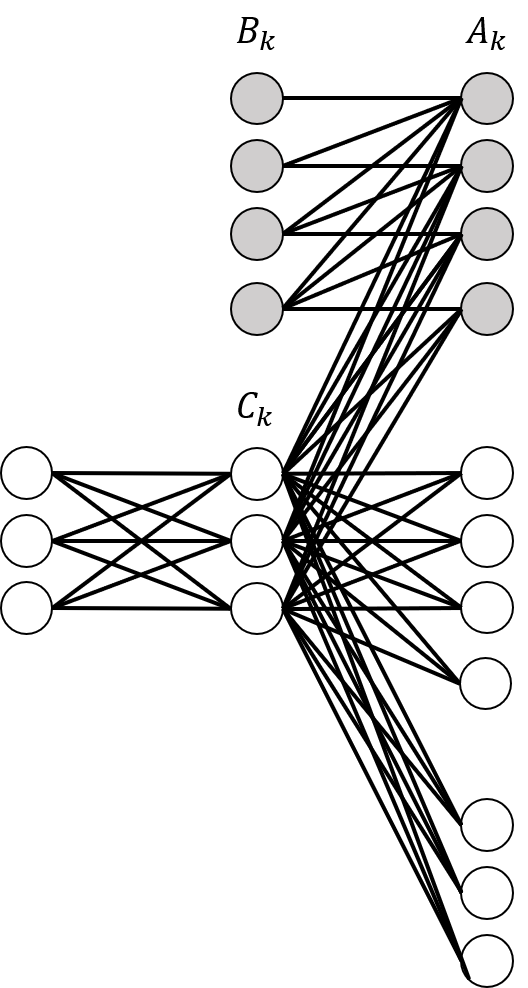}
        \caption{}
        \label{fig:construction-3}
    \end{subfigure}
    \begin{subfigure}[H]{.3\textwidth}
        \includegraphics[scale = 0.5]{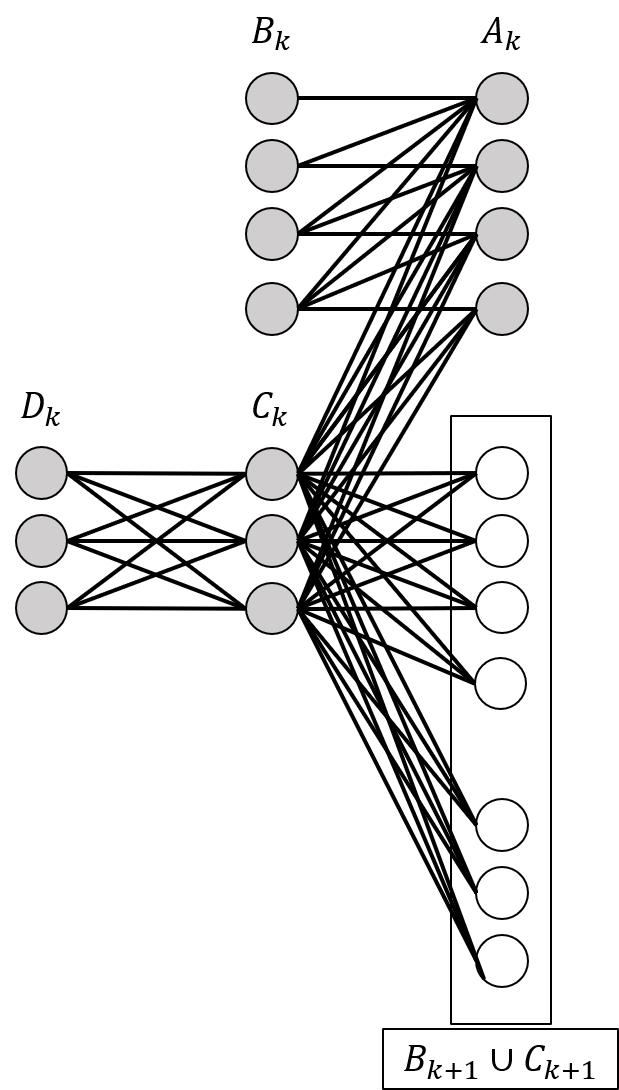}
        \caption{}
        \label{fig:construction-4}
    \end{subfigure}
    \caption{An illustration when $\alpha \cdot n =4$ and $(1-\alpha)\cdot n=3$. The vertices become grey means that they have already left.}
    \label{fig:construction}
\end{figure}

\paragraph{Competitive Ratio.}

First, we observe that the maximum matching of the graph has size $n\ell$. Indeed, the graph has a perfect matching by matching $A_k$ with $B_k$ and matching $C_k$ with $D_k$ for all $k \in [\ell]$.

Next, we shall assume that the algorithm only makes matching decisions on the departures of vertices, and the algorithm would like to fully match a vertex unless there is no feasible neighbor. This assumption is without loss of generality by a simple exchange argument (see e.g. \cite{jacm/HuangKTWZZ20}).

A crucial property of our construction is that on the departure of each vertex, all its remaining neighbors in the graph are indistinguishable. Indeed, on the departure of each $a_{ki}$, all its neighbors $\{b_{kj}\}_{j \ge i} \cup C_k$ are indistinguishable; on the departure of each $c_{ki}$, all its neighbors $D_k \cup B_{k+1} \cup C_{k+1}$ are indistinguishable. 
Applying a similar argument as~\cite{orl/EcklKLS21}, we can then, without loss of generality, restrict ourselves to the Water-Filling algorithm that treats all neighbors symmetrically. That is, on the departure of each vertex, we match its remaining portion equally to all its neighbors. A more formal construction need to decide the identity of each $b_{ki}$ after $a_{ki}$'s departure and decide the identities of $D_k$ after $C_k$'s departure. We omit the detailed derivation here.

Based on the discussion above, we are left to calculate the performance of Water-filling on the constructed instance. Since we treat all vertices in $B_k \cup C_k$ equally at the $(k-1)$-th stage, all vertices of $B_k \cup C_k$ would have the same water level right before the $k$-th stage. Let $x_{k}$ be value of the water level. Initially, we have $x_1 = 0$. We derive the recurrent formula of $x_{k+1}$ in the following.

On the deadline of the vertex $a_{ki}$ of $A_k$, it has $|\{b_{kj}\}_{j \ge i} \cup C_k| = n-i+1$ number of neighbors. Hence, the increment of the water level of $C_k$ equals $\frac{1}{n-i+1}$ on the departure of $a_{ki}$. Then, after the departure of the last vertex of $A_k$, the water level of $C_k$ becomes
\[
y_k = x_{k} + \frac{1}{n} + \frac{1}{n-1} + \cdots + \frac{1}{(1-\alpha)n + 1}~. 
\]
Next, on the departure of $C_k$, each vertex matches its remaining portion to all vertices in $D_k \cup B_{k+1} \cup C_{k+1}$ equally. Thus, we have
\[
x_{k+1} = (1 - y_k) \cdot \frac{1-\alpha}{2-\alpha} = \left(1 - \left(x_k + \frac{1}{n} + \frac{1}{n-1} + \cdots + \frac{1}{(1-\alpha)n+1}\right)\right) \cdot \frac{1-\alpha}{2-\alpha}~.
\]
Fix $n$ and let the number of stage $k$ goes to infinity. We have that the sequence $\{x_k\}$ converges to the solution $x^*$ of the following equation:
\[
x = \left(1 - \left(x + \frac{1}{n} + \frac{1}{n-1} + \cdots + \frac{1}{(1-\alpha)n+1}\right)\right) \cdot \frac{1-\alpha}{2-\alpha}~.
\]
Moreover, the total matching made in stage $k$ equals 
\[
1 \cdot \alpha n + (1-y_k) \cdot (1-\alpha) n = \left(\alpha + (2-\alpha)\cdot x_{k+1}\right) \cdot n,
\]
where the term $1 \cdot \alpha n$ corresponds to the matchings made on the departures of $A_k$ and the term $(1-y_k)\cdot(1-\alpha) n$ corresponds to the matchings made on the departures of $C_k$.
Hence, when $\ell \to \infty$, the ratio between the matching produced by Water-filling and the optimal matching is
\[
\lim_{k \to \infty} \frac{\left(\alpha + (2-\alpha)\cdot x_{k+1}\right) \cdot n}{n} = \alpha + (2-\alpha)\cdot x^*~.
\]
Finally, let $n$ goes to infinity, we have $x^*$ satisfies that
\[
x^* = \left( 1-x^*+\ln(1-\alpha) \right) \cdot \frac{1-\alpha}{2-\alpha} \iff x^*= \frac{1-\alpha}{3-2\alpha} \cdot (1+\ln(1-\alpha)) ~,
\]
where we use the fact that 
\[
\lim_{n\to \infty}\left( \frac{1}{n} + \frac{1}{n-1} + \cdots + \frac{1}{(1-\alpha)n+1} \right) = -\ln(1-\alpha)~.
\]
By optimizing $\alpha$, we conclude an upper bound of
\[
\min_{\alpha} \left\{ \alpha+(2-\alpha)\cdot x^* \right\} = \min_{\alpha} \left\{\alpha+\frac{(2-\alpha)(1-\alpha)}{3-2\alpha} \cdot (1+\ln(1-\alpha)) \right\} \approx 0.6132,
\]
when $\alpha \approx 0.43$.

\section{Improved Hardness for General Vertex Arrival}
\label{sec:general_upper}

In this section, we show an improved upper bound for the fractional online matching problem with general vertex arrival. Previously, the best known upper bound is $0.592$ by Buchbinder et al.~\cite{algorithmica/BuchbinderST19}. 

\begin{theorem}
\label{thm:general_upper}
No algorithm is strictly better than $0.584$-competitive for fractional online matching with general vertex arrivals, even for bipartite graphs.
\end{theorem}

\paragraph{Hard Instance.} Our construction includes $3$ stages. 
Suppose there exists a $\Gamma$-competitive algorithm.
We shall fix an arbitrary algorithm with competitive ratio $\Gamma$ and construct an instance that depends on the decision made by the algorithm. In the following, we describe the behavior of the algorithm and our construction at the same time.
Refer to Figure~\ref{fig:general}.

\begin{itemize}
    \item (Fig~\ref{fig:general-1}) The first stage involves $nk$ number of vertices that arrive in a sequence, that form a complete bipartite graph with an even size on the two sides of the graph. Denote these vertices by $A$. We denote the average matched portion of these vertices after the first stage by $\gamma$. If $\gamma < \Gamma$, the instance ends.
    \item (Fig~\ref{fig:general-2}) The second stage is consisted of $n$ steps. At step $i$, the identities of $\{A_{j}\}_{j<i}$ has been determined. A set of $k$ vertices $B_i$ arrive and it forms a complete bipartite graph with the vertices $A - \cup_{j<i}A_j$. The algorithm makes matching decision between $B_i$ and $A - \cup_{j<i}A_j$. Afterwards, let $A_i$ be the least $k$ matched vertices among $A - \cup_{j<i}A_j$. We use $\alpha_i,\beta_i$ to denote the average matched portion of $A_i,B_i$ after the second stage.
    \item (Fig~\ref{fig:general-3}) Let $I=\{i \in [n]:1-e^{\alpha_i-1} + 1 - e^{\beta_i-1} < \Gamma\}$.
    \item (Fig~\ref{fig:general-4}) For each $i \in I$, let there be $k$ vertices $C_i$ and an upper triangle between $C_i$ and $A_i$. The vertices in $C_i$ comes in a sequence. Upon the arrival of the $j$-th vertex $c_{ij}$, the identities of $\{a_{i\ell}\}_{\ell<j}$ has been determined. $c_{ij}$ is connected to all vertices of $A_i - \{a_{i\ell}\}_{\ell<j}$ and the algorithm makes matching decisions for these edges. Afterwards, let $a_{ij}$ be the least matched vertex among $A_i - \{a_{i\ell}\}_{\ell<j}$. Similarly, let there be $k$ vertices $D_i$ and an upper triangle between $D_i$ and $B_i$.
\end{itemize}

\begin{figure}
    \center
    \begin{subfigure}[H]{.15\textwidth}
        \includegraphics[scale = 0.5]{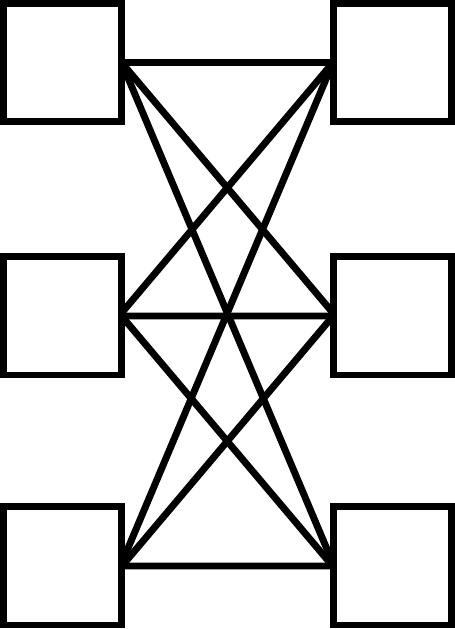}
        \caption{}
        \label{fig:general-1}
    \end{subfigure}
    \begin{subfigure}[H]{.15\textwidth}
        \includegraphics[scale = 0.5]{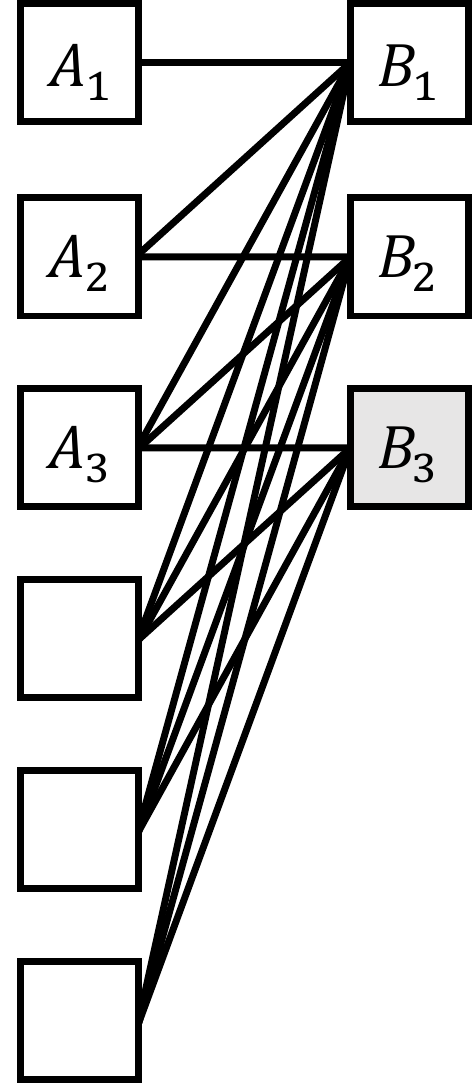}
        \caption{}
        \label{fig:general-2}
    \end{subfigure}
    \begin{subfigure}[H]{.17\textwidth}
        \includegraphics[scale = 0.5]{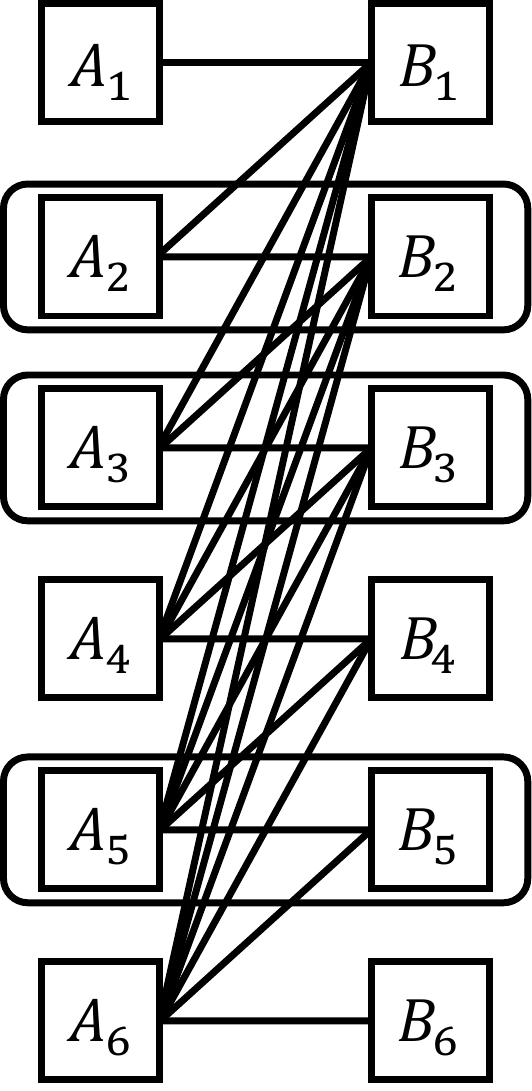}
        \caption{}
        \label{fig:general-3}
    \end{subfigure}
    \begin{subfigure}[H]{.4\textwidth}
        \includegraphics[scale = 0.5]{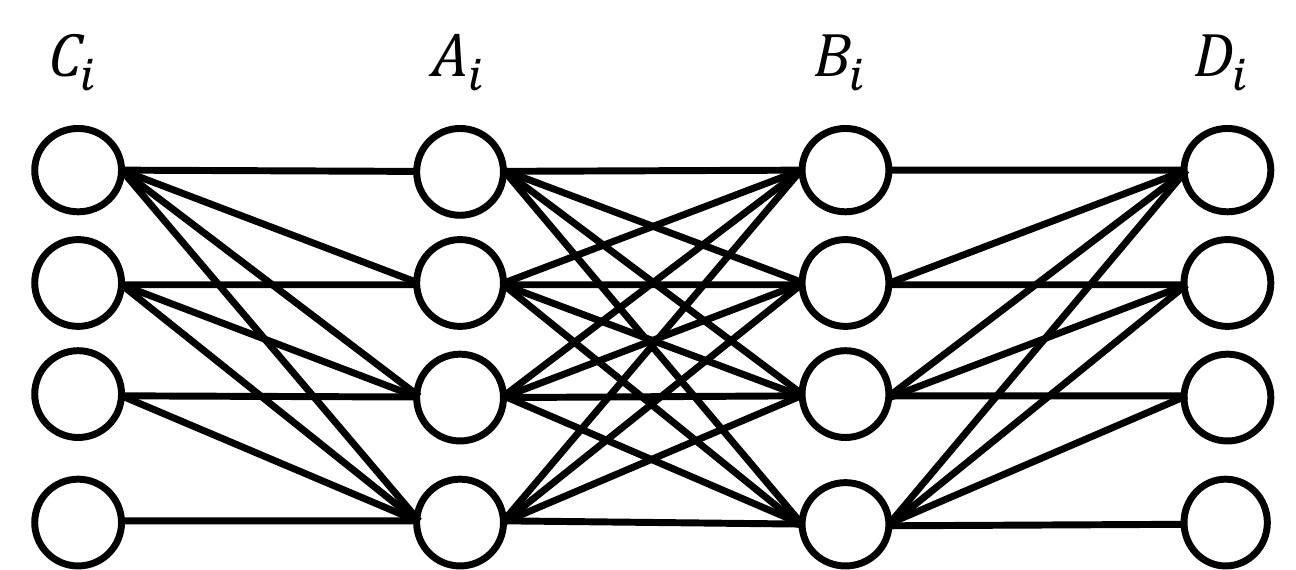}
        \caption{}
        \label{fig:general-4}
    \end{subfigure}
    \caption{An illustration when $n=6$ and $k=4$. Each square denotes a set of $k=4$ vertices. An edge between two squares means a complete bipartite graph in between the two squares. a) $24$ vertices arrive in a sequence and they form a complete bipartite graph. b) $B_3$ arrives and the algorithm makes matching decisions. Afterwards, the identity of $A_3$ is determined. c) $I = \{2,3,5\} = \{i:1-e^{a_i-1} + 1 - e^{b_i-1} < \Gamma\}$. d) Two upper triangles for each $i \in I$.}
    \label{fig:general}
\end{figure}

\paragraph{Competitive Ratio.} 
We first derive the stated competitive ratio of any algorithm for our instance, by making the following intuitive assumptions on the behavior of the algorithm. A formal proof of the theorem is provided in Appendix~\ref{app:general_upper}. 
\begin{itemize}
\item The algorithm matches evenly for all vertices in $A$ by $\gamma$ fraction at the first stage. Intuitively, since all vertices are symmetric at the first stage, the algorithm should match them evenly.
\item On the arrival of each $B_i$, the algorithm matches all remaining vertices, i.e. $A - \cup_{j <i} A_j$ evenly. This is again because all the neighbors of $B_i$ are symmetric on the arrival of $B_i$.
\item Last, the most crucial assumption we made is that whenever $A_i$ is not fully matched, i.e. $\alpha_i < 1$, on the arrival of $B_i$, we would keep matching until $1-e^{\alpha_i-1} + 1 - e^{\beta_i-1} = \Gamma$, where $\alpha_i,\beta_i$ denote the average matched portion of $A_i,B_i$ respectively. The value $1-e^{\alpha_i-1} + 1 - e^{\beta_i-1}$ is monotonically increasing as we continue matching $B_i$ to its neighbors. This assumption states that the algorithm should stop matching when this value is smaller than $\Gamma$.
\end{itemize}

The third assumption is the most crucial to achieve the improved upper bound of online matching with general vertex arrival. And this is due to the design of the third stage, which we explain in detail in the following. 

Recall the graph structure of the third stage. The upper triangle construction is essentially the same construction by Karp et al.~\cite{stoc/KarpVV90}, for proving the $1-\frac{1}{e}$ upper bound for the classical online bipartite matching problem. The only difference is that we have an average matched portion $\alpha_i,\beta_i$ of $A_i, B_i$ before the third stage. We analyze the best possible performance of any algorithm at the third stage, that can be viewed as a straightforward generalization of the $1-1/e$ hardness result of Karp et al.~\cite{stoc/KarpVV90}. We remark that when $\alpha_i=0$, our bound in the following lemma equals $1-\frac{1}{e}$. The formal proof of the lemma is deferred to Appendix~\ref{app:general_upper}.
\begin{lemma}
For each $i \in I$, the total matching made on the arrival of $C_i$ is at most $\left( 1-e^{\alpha_i-1} + o(1) \right) \cdot k$ and the total matching made on the arrival of $D_i$ is at most $\left( 1-e^{\beta_i-1} + o(1) \right) \cdot k$.
\end{lemma}

In other words, if an algorithm matches too much between $A_i$ and $B_i$ on the arrival of $B_i$, the arrivals of $C_i, D_i$ from the third stage would penalize the algorithm. Observe that after the arrival of each $C_i,D_i$, the size of the maximum matching increases by $k$, while the algorithm matches at most $(1-e^{\alpha_i-1} + 1-e^{\beta_i-1}) \cdot k$ by the above lemma. Since we are aiming for a competitive ratio of $\Gamma$, it suggests the algorithm to stop matching $B_i$ further if $1-e^{\alpha_i-1} + 1 - e^{\beta_i-1} = \Gamma$. This is exactly the third assumption we made above.  

Now we are ready to calculate the best possible competitive ratio $\Gamma$. We derive the claimed $\Gamma=0.584$ with computer assistance. Concludes all the assumptions above, we can solve all $\alpha_i$ and $\beta_i$ to depict the best algorithm's performance in the following Algorithm~\ref{alg:solve-ratio}.
\begin{algorithm}[H]
	\caption{Solve $\alpha$ an $\beta$}
	\label{alg:solve-ratio}
	\begin{algorithmic}
	    \State $\alpha_0=\gamma$.
	    \For{Each $i\in[n]$}
	       \State Solve $x$ for $e^{\alpha_{i-1} + \frac{x}{n-i+1} - 1} + e^{x-1} = 2-\Gamma$.
	       \State $\beta_i=\min\{x,(n-i+1)(1-\alpha_i)\}$.
	       \State $\alpha_i=\alpha_{i-1}-1+\frac{\beta_i}{n-i+1}$.
	    \EndFor
	\end{algorithmic}
\end{algorithm}

Fixing $n=500$, the algorithm must start with a $\gamma \ge \Gamma$ after the first stage. Then, the competitive ratio of the algorithm after the second stage equals:
\[
r(\gamma) = \frac{\left( n\gamma+2\sum_{i=1}^{n} \beta_i \right) \cdot k}{2n \cdot k} = \frac{n\gamma + 2\sum_{i=1}^{n} \beta_i}{2n}.
\]

We search for every possible value of $\gamma \ge \Gamma$ (discretized) and plot the following figure for the function $r(x)$.
\begin{figure}[H]
    \centering
    \includegraphics[scale = 0.5]{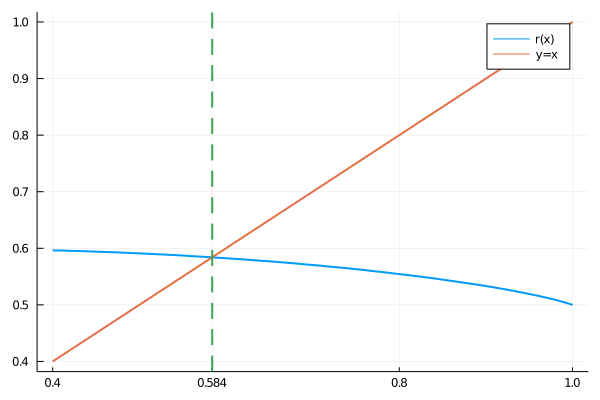}
    \caption{Function $r(x)$ and $y=x$}
    \label{fig:general-upper-func}
\end{figure}
Observe that $r(x) < \Gamma$ for all $x\ge \Gamma$. We conclude that there is no algorithm has a competitive ratio better than $\Gamma=0.584$.


\bibliography{matching}
\bibliographystyle{plain}

\appendix
\section{Factor Revealing LP}
\label{app:factor_lp}
In this section, we explain how we use the factor revealing LP to construct the feasible $h$ function to conclude the proof of Lemma~\ref{lem:opth}. Recall that there are two constraints between the non-decreasing function $h$ and the targeting competitive ratio $\Gamma = 0.6$. Let
\begin{align*}
    &\Phi_1(\tau_v,\theta_v) = \theta_v \cdot h(\tau_v, \theta_v) - \int_{\tau_v}^{\theta_v} h(\tau_v, y_v) dy_v + 1-\theta_v; \\ 
    &\Phi_2(\tau_u,\theta_u,\tau_v,\theta_v) = \left(
    \begin{aligned}
        &\theta_u \cdot h(\tau_u, \theta_u) - \int_{\tau_u}^{\theta_u} h(\tau_u, y_u) dy_u + \theta_v \cdot h(\tau_v, \theta_v)\\
        & - \int_{\tau_v}^{\theta_v} h(\tau_v, y_v) dy_v + (1-h(\tau_u, \theta_u)) \cdot (1-\theta_v)
    \end{aligned}
    \right).
\end{align*}
The two constraints are defined as follows: 
\begin{align}
	& \forall 0\le \tau_v \le \theta_v \le 1: &\Phi_1 \ge \Gamma; \tag{\ref{eqn:h-feasible-1}} \\
    & \forall 0\le \tau_u \le \theta_u \le 1, 1-\theta_u\le \tau_v \le \theta_v \le 1 :\quad  & \Phi_2 \ge \Gamma; \tag{\ref{eqn:h-feasible-2}}
\end{align}

\paragraph{The Factor-revealing Lp and The Construction of $h$} Fix an integer $n$, we let $1/n$ be the step size, and $[0,1]_n = \{\frac{i}{n}:i = 0 ... n\}$ denote the set of $n+1$ multiples of the step size $1/n$, between $0$ and $1$. The LP contains $(n+1)^2$ variables, which also represent the function value of $h(0,0), h(0,1/n), ..., h(0,1), h(1/n,0), ..., h(1,1)$. We remark that the remaining of $h$ will be later constructed by these variables via linear interpolation. In particular, consider a real number $x\in [0,1]$, we define $\bar{x}$ be the largest value inside $[0,1]_n$ that not exceed $x$ and $\hat{x} = \bar{x} + \frac{1}{n}$. We also define $z(x) = n \cdot (x - \bar{x})$, which also means $x = \bar{x} + \frac{z(x)}{n}$ and $x = \hat{x} - \frac{(1-z(x))}{n}$. Then for those point $(\tau,\theta)$ not inside $[0,1]_n \times [0,1]_n$, $h$ is defined to be:  
$$
h(\tau,\theta) = \left(
    \begin{aligned}
        &(1-z(\tau))\cdot (1-z(\theta)) \cdot h(\bar{\tau},\bar{\theta}) + (1-z(\tau))\cdot z(\theta) \cdot h(\bar{\tau},\hat{\theta}) \\
        &z(\tau)\cdot (1-z(\theta)) \cdot h(\hat{\tau},\bar{\theta}) + z(\tau)\cdot z(\theta) \cdot h(\hat{\tau},\hat{\theta})
    \end{aligned}
\right).
$$
Following this construction of $h$, $\forall \tau,~\theta \in [0,1]_n$, we have 
$$
\int_{\tau}^{\theta} h(\tau,y) dy = \sum_{\tau\leq y < \theta,~ y \in [0,1]_n} \frac{h(y) + h(y+\frac{1}{n})}{2n}.
$$
It means both when the variables (i.e., $\tau_u$, $\tau_v$, $\theta_u$, and $\theta_v$) are chosen from $[0,1]_n$, both $\Phi_1$ and $\Phi_2$ are linear. It allows us to present the LP formulation below, for convenience, we let $h(\tau,1+1/n) = 1$ for any $\tau \in [0,1]_n$. 

\begin{align}
    \text{maximize}\qquad & \Gamma \nonumber \\
    \text{subject to}
    \qquad & \Gamma \leq \Phi_1(\tau_v,\theta_v) - \frac{5}{2n^2}, 
    &\forall \tau_v \leq \theta_v \in [0,1]_n & \label{constraint:v-arrive-first}& \\
    & \Gamma \leq \Phi_2(\tau_u,\theta_u,\tau_v,\theta_v) - \frac{5}{n^2}, 
    &\begin{aligned}
        \forall \tau_v \leq \theta_v \in [0,1]_n, \\
        \forall 1-\theta_u\le \tau_v \le \theta_v \in [0,1]_n 
    \end{aligned} &
    \label{constraint:u-arrive-first} &\\
& 
\begin{aligned}
    &h(\tau,0) = 0, \\
    &h(\tau,1) = h(\tau, 1+\frac{1}{n}) = 1,
\end{aligned}\quad 
  &\forall \tau \in [0,1])_n &&\\ 
& h(\tau,y) < h(\tau, y+\frac{1}{n}),
  & \forall \tau \in [0,1], ~\forall y\in[0,1]_n < 1 &\tag{Monotonicity} &\\
&\begin{aligned}
 & \textstyle h(\tau, y) \geq h(\tau, y+\frac{1}{n}) + \frac{4}{n}, \\
 & \textstyle h(y, \theta) \geq h(y+\frac{1}{n}, \theta) + \frac{4}{n},  
\end{aligned} 
& \begin{aligned}
    &\forall \tau \in [0,1], ~\forall y\in[0,1]_n < 1 \\
    &\forall \theta \in [0,1], ~\forall y\in[0,1]_n < 1 
\end{aligned}
&\tag{Lipschitzness} &
\end{align}

We use the Gurobi LP solver to solve the above LP and get $\Gamma >0.6$ when we fix $n=100$. In the following part, we prove that the optimal solution of the above LP can construct a feasible $h$ function when we fix $\Gamma = 0.6$, i.e., prove Eqn~\eqref{eqn:h-feasible-1} and Eqn~\eqref{eqn:h-feasible-2}, so to derive the $0.6$ competitive ratio. 

These constraints directly appears in the LP when the parameters (i.e., $\tau_u$, $\tau_v$, $\theta_u$, and $\theta_v$) are all inside $[0,1]_n$, it remains to consider the case when some parameters are outside $[0,1]_n$. We first define following two functions: 

\begin{align*}
&\tilde{\Phi}_1(\tau_v,\theta_v)= \left(
\begin{aligned}
    & (1-z(\tau_u))\cdot(1-z(\theta_u)) \cdot \Phi_1(\bar{\tau_u},\bar{\theta_u}) 
    + (1-z(\tau_u))\cdot z(\theta_u) \cdot \Phi_1(\bar{\tau_u},\hat{\theta_u}) \\
    &+ z(\tau_u) \cdot (1-z(\theta_u)) \cdot \Phi_1(\hat{\tau_u},\bar{\theta_u}) 
    + z(\tau_u) \cdot z(\theta_u) \cdot \Phi_1(\hat{\tau_u},\hat{\theta_u})
\end{aligned}\right).
\\
&\tilde{\Phi}_2(\tau_u,\theta_u,\tau_v,\theta_v)= \left(
\begin{aligned}
    & (1-z(\tau_u))\cdot(1-z(\theta_u))\cdot(1-z(\tau_v))\cdot(1-z(\theta_v)) \cdot \Phi_2(\bar{\tau_u},\bar{\theta_u},\bar{\tau_v},\bar{\theta_v}) \\
    &+ (1-z(\tau_u))\cdot(1-z(\theta_u))\cdot(1-z(\tau_v))\cdot z(\theta_v) \cdot \Phi_2(\bar{\tau_u},\bar{\theta_u},\bar{\tau_v},\hat{\theta_v}) \\
    &+ (1-z(\tau_u))\cdot(1-z(\theta_u))\cdot z(\tau_v) \cdot (1-z(\theta_v)) \cdot \Phi_2(\bar{\tau_u},\bar{\theta_u},\hat{\tau_v},\bar{\theta_v}) \\
    &+ \dots \\
    &+ z(\tau_u) \cdot z(\theta_u) \cdot z(\tau_v) \cdot z(\theta_v) \cdot \Phi_2(\hat{\tau_u},\hat{\theta_u},\hat{\tau_v},\hat{\theta_v})
\end{aligned}
\right).
\end{align*}
It's easy to see that $\tilde{\Phi}_1$ and $\tilde{\Phi}_2$ is no less than $\Gamma + \frac{5}{2n^2}$ and $\Gamma + \frac{5}{n^2}$, because they are the affine combination of the terms appear in the LP constraints \eqref{constraint:v-arrive-first} and \eqref{constraint:u-arrive-first}. The key point we claim is the following two lemmas:
\begin{lemma}
\label{lem:phi1-bound}
    $\Phi_1(\tau_u,\theta_u,\tau_v,\theta_v) \geq \tilde{\Phi}_1(\tau_u,\theta_u,\tau_v,\theta_v) - \frac{5}{2n^2} \geq \Gamma$.
\end{lemma}
\begin{lemma}
\label{lem:phi2-bound}
    $\Phi_2(\tau_u,\theta_u,\tau_v,\theta_v) \geq \tilde{\Phi}_2(\tau_u,\theta_u,\tau_v,\theta_v) - \frac{5}{n^2} \geq \Gamma$.
\end{lemma}
To prove these two lemmas, we define one crucial function below: 
\begin{align*}
&H(\tau,\theta) = \theta_v h(\tau,\theta) - \int_{\tau}^{\theta} h(\tau,y) dy; \\
&\tilde{H}(\tau,\theta) = \left(
\begin{aligned}
    &(1-z(\tau))(1-z(\theta))H(\bar{\tau},\bar{\theta}) 
    + (1-z(\tau))\cdot z(\theta) H(\bar{\tau},\hat{\theta}) \\
    &+ z(\tau)\cdot(1-z(\theta))H(\hat{\tau},\bar{\theta}) 
    + z(\tau)\cdot z(\theta) H(\hat{\tau},\hat{\theta})
\end{aligned}
\right).
\end{align*}
We claim the following lemma, 

\begin{lemma}
    $H(\tau,\theta) \geq \tilde{H}(\tau,\theta) - \frac{5}{2n^2}$.
\end{lemma}
\begin{proof}
    We first restrict $\tau \in [0,1]_n$ but relax $\theta$ to be in $[0,1]$, and prove that 
    $$
    H(\tau,\theta) \geq (1-z(\theta))\cdot H(\tau,\bar{\theta}) + z(\theta)\cdot H(\tau,\hat{\theta}) - \frac{1}{n^2}.
    $$
    Recall that $H(\tau,\theta) = \theta h(\tau,\theta) - \int_{\tau}^{\theta} h(\tau,y) dy$, consider the first term, we have 
    \begin{align*}
        & \theta h(\tau,\theta) - (1-z(\theta))\cdot \bar{\theta} h(\tau,\bar{\theta}) - z(\theta)\cdot \hat{\theta} h(\tau,\hat{\theta}) \\ 
        = & \theta\cdot(1-z(\theta))\cdot(\tau,\bar{\theta}) + \theta\cdot z(\theta)h(\tau,\hat{\theta}) - (1-z(\theta))\cdot \bar{\theta} h(\tau,\bar{\theta}) - z(\theta)\cdot \hat{\theta} h(\tau,\hat{\theta}) \\
        = & \frac{z(\theta)}{n} \cdot (1-z(\theta)) \cdot h(\tau,\bar{\theta}) - \frac{1-z(\theta)}{n} \cdot z(\theta) \cdot h(\tau,\hat{\theta}) \\
        = & \frac{z(\theta)(1-z(\theta))}{n} \cdot (h(\tau,\bar{\theta}) -h(\tau,\hat{\theta})) \\ 
        \geq & - \frac{1}{n^2}. \qquad \text{(By the Lipschitzness of $h$)}
    \end{align*}
    For the other term, we have
    \begin{align*}
        & - \int_{\tau}^{\theta} h(\tau,y)dy + (1-z(\theta)) \cdot \int_{\tau}^{\bar{\theta}} h(\tau,y)dy + z(\theta) \cdot \int_{\tau}^{\hat{\theta}} h(\tau,y)dy \\
        = & - \int_{\bar{\theta}}^{\theta} h(\tau,y)dy + z(\theta) \cdot \int_{\hat{\theta}}^{\bar{\theta}} h(\tau,y)dy \\
        \geq & 0. \qquad \text{(By the Monotonicity of $h$)}
    \end{align*}
    Thus far, we conclude the first claim. Then, we restrict $\theta \in [0,1]_n$, and allow $\tau$ to be in $[0,1]$. We prove that 
    $$
    H(\tau,\theta) \geq (1-z(\tau))\cdot H(\bar{\tau},\theta) + z(\tau)\cdot H(\hat{\tau},\theta) - \frac{3}{2n^2}.
    $$
    Similarly, consider the first term, we have
    \begin{align*}
        & \theta h(\tau,\theta) - (1-z(\tau))\cdot \theta h(\bar{\tau},\theta) - z(\tau)\cdot \theta h(\hat{\tau},\theta) = 0.
    \end{align*}
    For the second them, we have 
    \begin{align*}
        & - \int_{\tau}^{\theta} h(\tau,y)dy + (1-z(\tau)) \cdot \int_{\bar{\tau}}^{\theta} h(\bar{\tau},y)dy + z(\tau) \cdot \int_{\hat{\tau}}^{\theta} h(\hat{\tau},y)dy \\
        = & - (1-z(\tau)) \cdot \int_{\tau}^{\theta} h(\bar{\tau},y)dy - z(\tau) \cdot \int_{\tau}^{\theta} h(\hat{\tau},y)dy + (1-z(\tau)) \cdot \int_{\bar{\tau}}^{\theta} h(\bar{\tau},y)dy + z(\tau) \cdot \int_{\hat{\tau}}^{\theta} h(\hat{\tau},y)dy \\
        = & (1-z(\tau)) \cdot \int_{\bar{\tau}}^{\tau} h(\bar{\tau},y)dy  - z(\tau) \cdot \int_{\tau}^{\hat{\tau}} h(\hat{\tau},y)dy \\
        = & (1-z(\tau)) \cdot \frac{z(\tau)}{2n} \cdot(h(\bar{\tau},\bar{\tau}) + h(\bar{\tau},\tau))  - z(\tau) \cdot \frac{1-z(\tau)}{2n} \cdot(h(\hat{\tau},\tau) + h(\hat{\tau},\hat{\tau}))\\
        \geq & -\frac{z(\tau)\cdot(1-z(\tau))}{2n}\cdot \frac{12}{n} \qquad \text{(By the Lipschitzness of $h$)} \\
        \geq & \frac{3}{2n^2} 
    \end{align*}
\end{proof}

Finally, we conclude the proof of Lemma~\ref{lem:phi1-bound} and Lemma~\ref{lem:phi2-bound}.
\begin{proof}[Proof of Lemma~\ref{lem:phi1-bound}]
\begin{align*}
    \Phi_1(\tau_v,\theta_v)
    = H(\tau_v,\theta_v) + 1 - \theta_v
    \geq \tilde{H}(\tau_v,\theta_u) + 1 - \theta_v - \frac{5}{2n^2} \geq \tilde{\Phi}_1(\tau_u,\theta_u,\tau_v,\theta_v) - \frac{5}{2n^2} \geq \Gamma.
\end{align*}
\end{proof}
\begin{proof}[Proof of Lemma~\ref{lem:phi2-bound}]
\begin{align*}
    &\theta_u \cdot h(\tau_u, \theta_u) - \int_{\tau_u}^{\theta_u} h(\tau_u, y_u) dy_u + \theta_v \cdot h(\tau_v, \theta_v) - \int_{\tau_v}^{\theta_v} h(\tau_v, y_v) dy_v + (1-h(\tau_u, \theta_u)) \cdot (1-\theta_v) \\
    =& H(\tau_u,\theta_u) + H(\tau_v,\theta_v) + (1-h(\tau_u,\theta_u))\cdot(1-\theta_v)\\
    \geq& \tilde{H}(\tau_u,\theta_u) + \tilde{H}(\tau_v,\theta_v) + (1-h(\tau_u,\theta_u))\cdot(1-\theta_v) - \frac{10}{n^2}\\
    =& \tilde{H}(\tau_u,\theta_u) + \tilde{H}(\tau_v,\theta_v) + 1-h(\tau_u,\theta_u) -\theta_v + \theta_v h(\tau_u,\theta_u) - \frac{5}{n^2} \\
    =& \tilde{\Phi}_2(\tau_u,\theta_u,\tau_v,\theta_v) - \frac{5}{n^2} \geq \Gamma.
\end{align*}
\end{proof}

\section{Formal Proof of Theorem~\ref{thm:general_upper}}
\label{app:general_upper}

Recall that $\gamma$ is the average matched portion of vertices in $A$ after the first stage; $\alpha_i, \beta_i$ equal the average matched portion of vertices in $A_i, B_i$ respectively, after the second stage. We first have the following observation from our construction.
\begin{lemma}
\label{lem:constraint}
$\{\alpha_i\},\{\beta_i\},\gamma$ satisfy the following conditions:
\begin{itemize}
\item $\forall i \in [n-1], \quad \sum_{j=1}^{i-1} \alpha_j + (n-i+1) \cdot \alpha_i \le \sum_{j=1}^{i} \beta_j + n \cdot \gamma$;
\item $\sum_{i=1}^{n} \alpha_i = \gamma \cdot n + \sum_{i=1}^{n} \beta_i$.
\end{itemize}
\end{lemma}
\begin{proof}
For each $i\in[n]$, consider the moment after $B_i$'s arrival. According to the definition of $A_i$, they have the least average matched portion among all vertices in $\cup_{j=i}^{n} A_j$. Thus, the total matched portion of all vertices in $A$ is at least $\left( \sum_{j=1}^{i-1}\alpha_j + (n-i+1) \cdot \alpha_i \right) \cdot k$. On the other hand, the total matched portion of all vertices in $A$ should equal the total matched portion of $\cup_{j=1}^{i} B_j$ plus the total matched portion from the first stage, which is $\left(\sum_{j=1}^{i} \beta_j + n \cdot \gamma \right)\cdot k$. This gives the first family of inequalities.
Moreover, after the second stage (i.e., the arrival of $B_n$), the inequality becomes an equality since the total matched portion of $A$ equals $\sum_{i=1}^{n} \alpha_i$, which gives the second equation.
\end{proof}

Next, we study the optimal behavior of any algorithm at the third stage.
\begin{lemma}
\label{lem:triangle}
For each $i \in I$, the total matching made on the arrival of $C_i$ is at most $\left( 1-e^{\alpha_i-1} + o(1) \right) \cdot k$ and the total matching made on the arrival of $D_i$ is at most $\left( 1-e^{\beta_i-1} + o(1) \right) \cdot k$.
\end{lemma}
\begin{proof}
Let $x_j$ be the matched portion of vertex $a_{ij}$ for each $j\in[k]$. Similar to the analysis above, consider the moment after $a_{ij}$'s arrival. By our construction, $a_{ij}$ is the least matched vertex among $\{a_{i\ell}\}_{\ell=j}^{k}$. Therefore, $\sum_{\ell=1}^{j-1} x_\ell + (k-j+1) \cdot x_j \le \alpha_i \cdot k + j$, here the right hand side comes from the fact that the average matched portion of $A_i$ equals $\alpha_i$ before the third stage and each $c_{i\ell} \in C_i$ can match at most one edge on its arrival.
Let $j^*$ be an index to be optimized later. Summing up the above inequality for $j \le j^*$ with an appropriate factor, we have
\begin{align*}
& \sum_{j=1}^{j^*} \frac{1}{(k-j)(k-j+1)} \cdot \left( \sum_{\ell=1}^{j-1} x_\ell + (k-j+1) \cdot x_j\right) \le \sum_{j=1}^{j^*} \frac{1}{(k-j)(k-j+1)} \cdot \left( \alpha_i \cdot k + j \right) \\
\Rightarrow & \sum_{j=1}^{j^*} \left( \frac{1}{k-j} + \sum_{\ell=j+1}^{j^*} \frac{1}{(k-\ell)(k-\ell+1)} \right) x_j \le k \cdot (\alpha_i+1) \cdot \left(\frac{1}{k-j^*} - \frac{1}{k}\right) - \sum_{j=1}^{j^*} \frac{1}{(k-j+1)} \\
\Rightarrow & \frac{1}{k-j^*} \cdot \sum_{j=1}^{j^*} x_j \le (\alpha_i+1) \cdot \frac{j^*}{k-j^*} - \ln \left(\frac{k}{k-j^*} \right) + o(1) \\
\Rightarrow & \sum_{j=1}^{j^*} x_j \le (\alpha_i+1) \cdot j^* - (k-j^*) \cdot \left( \ln \left(\frac{k}{k-j^*} \right) - o(1) \right) 
\end{align*}
For $j>j^*$, we apply the trivial bound of $x_j \le 1$. To sum up, we have
\[
\sum_{j=1}^{k} x_j \le (\alpha_i+1) \cdot j^* + (k-j^*) \cdot \left( 1 - \ln \left(\frac{k}{k-j^*} \right) + o(1) \right).
\]
Let $j^* = \left\lfloor (1-e^{\alpha_i-1})\cdot k \right\rfloor$, and when $k\to \infty$,
\begin{multline*}
 (\alpha_i+1) \cdot j^* + (k-j^*) \cdot \left( 1 - \ln \left(\frac{k}{k-j^*} \right) + o(1) \right) \\
= k \cdot \left((\alpha_i+1) (1-e^{\alpha_i-1}) + e^{\alpha_i-1} \cdot (1 - \ln (e^{1-\alpha_i})) + o(1) \right) = k \cdot \left(\alpha_i + 1 - e^{\alpha_i-1} + o(1) \right).
\end{multline*}
Finally, observe that the total matching made on the arrival of $C_i$ is
\[
\sum_{j=1}^{k} x_j - \alpha_i \cdot k \le \left(1 - e^{\alpha_i-1} + o(1) \right) \cdot k.
\]
Similarly, the total matching made on the arrival of $D_i$ is at most $\left( 1-e^{\beta_i-1} + o(1) \right) \cdot k$.
\end{proof}

According to the above discussion, we conclude that the optimal competitive ratio for our construction is upper bounded by the following optimization problem.
\begin{align}
\label{program:p1}
\max_{\Gamma,\gamma,\{\alpha_i\},\{\beta_i\}}: \quad  & \Gamma \tag{P1}\\
\label{eq:1}
\mbox{subject to :} \quad & \gamma \ge \Gamma; \\
\label{eq:2}
& \sum_{i=1}^{n} (\alpha_i + \beta_i) \ge 2 \Gamma n + 2 \sum_{i=1}^{n} \max(0,\Gamma - (1-e^{\alpha_i-1}) - (1-e^{\beta_i-1})); \\
\label{eq:3}
\forall i \in [n], \quad & \sum_{j=1}^{i-1} \alpha_j + (n-i+1) \cdot \alpha_i \le \sum_{j=1}^{i} \beta_j + n \cdot \gamma; \\
\label{eq:4}
	& \sum_{i=1}^{n} \alpha_i \ge \gamma \cdot n + \sum_{i=1}^{n} \beta_i; \\
\forall i \in [n], \quad & 0 \le \alpha_i, \beta_i \le 1 \\
& \{\alpha_i\} \text{ is increasing}
\end{align}

The first constraint~\eqref{eq:1} corresponds to the first stage. Note that if $\gamma < \Gamma$, our instance ends right after the first stage and the algorithm cannot be $\Gamma$ competitive. The second constraint~\eqref{eq:2} corresponds to the case when our instance ends after the third stage.
By Lemma~\ref{lem:triangle}, the total matched portion over all vertices of our algorithm is at most
\[
\left( \sum_{i=1}^{n} (\alpha_i + \beta_i) + 2 \cdot \sum_{i \in I} (1-e^{\alpha_i-1} + 1-e^{\beta_i-1}) + o(1) \right) \cdot k.
\]
On the other hand, there is a perfect matching of the graph, by matching every $A_i$ with $C_i$ and $B_i$ with $D_i$ for $i \in I$, and matching every $A_i$ with $B_i$ for $i \notin I$. That is, the optimal matching has size $k \cdot (2n + 2|I|)$, in terms of the number of matched vertices. Let $k$ tends to infinity. If an algorithm is $\Gamma$-competitive, it must satisfies that
\begin{align*}
& \sum_{i=1}^{n} (\alpha_i + \beta_i) + 2 \cdot \sum_{i \in I} (1-e^{\alpha_i-1} + 1-e^{\beta_i-1})  \ge \Gamma \cdot (2n + 2|I|) \\
 \iff & \sum_{i=1}^{n} (\alpha_i + \beta_i) \ge 2 \Gamma n + 2 \sum_{i=1}^{n} \max(0,\Gamma - (1-e^{\alpha_i-1}) - (1-e^{\beta_i-1})),
\end{align*}
due to our construction of $I = \{i \in [n]: 1-e^{\alpha_i-1} + 1-e^{\beta_i-1} < \Gamma\}$.
The third and fourth constraints~\eqref{eq:3} and \eqref{eq:4} are from Lemma~\ref{lem:constraint}. We are now left to solve the optimization. 

\paragraph{Solving the Optimization.} We shall prove that certain constraints of the above optimization are tight, that indeed justify the assumptions we made in Section~\ref{sec:general_upper}. 

\begin{lemma}
There exists an optimal solution to \eqref{program:p1} that $1-e^{\alpha_i-1} + 1-e^{\beta_i-1} \ge \Gamma$ for all $i\in[n]$.
\end{lemma}
\begin{proof}
If there exists an $i$ with $1-e^{\alpha_i-1} + 1-e^{\beta_i-1} < \Gamma$, we decrease the values of $\alpha_i, \beta_i$ by $\epsilon$. It is straightforward to check that all constraints of \eqref{eq:3},~\eqref{eq:4} become looser. For constraint~\eqref{eq:2}, observe that the LHS decreases by $2\epsilon$, while the RHS decreases by
\[
2 (e^{\alpha_i+\epsilon-1} - e^{\alpha_i-1} + e^{\beta_i+\epsilon-1} - e^{\beta_i-1}) \approx 2\epsilon \cdot (e^{\alpha_i-1} + e^{\beta_i-1}) > 2 \epsilon \cdot (2-\Gamma) > 2\epsilon, 
\]
where the last inequality follows from the fact that the solution must satisfy that $\Gamma < 1$. Therefore, after the modification to $\alpha_i, \beta_i$, all constraints become looser. We can thus keeps increasing $\alpha_i, \beta_i$ until the stated condition holds.
\end{proof}

The above lemma corresponds to a weaker version of the third assumption we made in Section~\ref{sec:general_upper}: $(1-e^{\alpha_{i}-1}) + (1-e^{\beta_i-1}) \ge \Gamma$ for all $i$. Recall that in Section~\ref{sec:general_upper}, we further assume that this inequality is tight if the constraint $\alpha_i < 1$ is not tight. We shall elaborate on it later.

Next, we observe that the variables $\beta_i$'s are decreasing in the optimal solution. 
\begin{lemma}
There exists an optimal solution to \eqref{program:p1} that $\{\beta_i\}$ is decreasing.
\end{lemma}
\begin{proof}
By reordering all $\beta_i$'s in a descending order, it is straightforward to check that all constraints become looser.	
\end{proof}

Therefore, it suffices to solve the following optimization problem.
\begin{align}
\label{program:p2}
\max_{\Gamma,\gamma,\{\alpha_i\},\{\beta_i\}}: \quad  & \Gamma \tag{P2}\\
\mbox{subject to :} \quad & \gamma \ge \Gamma; \\
& \sum_{i=1}^{n} (\alpha_i + \beta_i) \ge 2 \Gamma n; \\
\label{cons:gamma}
\forall i \in [n], \quad & (1-e^{\alpha_{i}-1}) + (1-e^{\beta_i-1})) \ge \Gamma; \\
\label{cons:even}
\forall i \in [n], \quad & \sum_{j=1}^{i-1} \alpha_j + (n-i+1) \cdot \alpha_i \le \sum_{j=1}^{i} \beta_j + n \cdot \gamma; \\
	& \sum_{i=1}^{n} \alpha_i \ge \gamma \cdot n + \sum_{i=1}^{n} \beta_i; \\
\forall i \in [n], \quad & 0 \le \alpha_i, \beta_i \le 1 \\
& \{\alpha_i\} \text{ is increasing}, \{\beta_i\} \text{ is decreasing}
\end{align}

Finally, we introduce two operations for modifying a feasible solution of \eqref{program:p2}. With these two operations, we can apply them alternatively so that the solution we obtain satisfies all the three assumptions of Section~\ref{sec:general_upper}.
\begin{lemma}
Let $i$ be the minimum index with $1-e^{\alpha_i-1} + 1-e^{\beta_i-1} > \Gamma$. If there exists $j >i$ with $\beta_j > 0$, we can construct another feasible solution so that $1-e^{\alpha_\ell-1} + 1-e^{\beta_\ell-1} = \Gamma$ for all $\ell \le i$.
\end{lemma}
\begin{proof}
Suppose there exists an $i$ with $(1-e^{\alpha_{i}-1}) + (1-e^{\beta_i-1})) > \Gamma$ and there exists a $\beta_j >0$ for some $j >i$.
We modify the solution by increasing $\beta_i$ by $\epsilon$ and decreasing $\beta_j$ by $\epsilon$ at the same time. It is straightforward to check all constraints become looser except the constraint $(1-e^{\alpha_{i}-1}) + (1-e^{\beta_i-1})) > \Gamma$.
\end{proof}

\begin{lemma}
Let $i$ be the minimum index with $\sum_{j=1}^{i-1} \alpha_j + (n-i+1) \cdot \alpha_i < \sum_{j=1}^{i} \beta_j + n \cdot \gamma$. We can construct another feasible solution so that $\sum_{j=1}^{i-1} \alpha_j + (n-i+1) \cdot \alpha_i = \sum_{j=1}^{i} \beta_j + n \cdot \gamma;$ for all $\ell \le i$.
\end{lemma}

\begin{proof}
Suppose constraint \eqref{cons:even} is not tight for some $i$, consider increasing $\alpha_i$ by an sufficiently small $\epsilon$ and decreasing $\beta_i$ by $\delta$ so that 
$1-e^{\alpha_i-1} + 1-e^{\beta_i-1}$ remains the same. Next, we decrease $\alpha_{i+1}$ and increase $\beta_{i+1}$. We claim that all constraints remain feasible.

Let $\alpha_i' = \alpha_i' + \epsilon, \alpha_{i+1}' = \alpha_{i+1} - \epsilon, \beta_i'=\beta_i - \delta$, and $\beta_{i+1}' = \beta_{i+1} + \delta$.
Then 
\begin{multline*}
\left(1-e^{\alpha_{i+1}'-1} + 1-e^{\beta_{i+1}'-1} \right) - \left(1-e^{\alpha_{i+1}-1} + 1-e^{\beta_{i+1}-1} \right) = e^{\alpha_{i+1}-1}-e^{\alpha_{i+1}'-1} + e^{\beta_{i+1}-1} - e^{\beta_{i+1}'-1} \\
= e^{\alpha_{i+1}' -1} (e^{\epsilon} - 1)  - e^{\beta_{i+1}-1}(e^{\delta} -1) \ge e^{\alpha_i -1} (e^{\epsilon} - 1)  - e^{\beta_i'-1}(e^{\delta} -1) = 0,
\end{multline*}
where the inequality holds by the monotonicity of $\{\alpha_j\}$ and $\{\beta_j\}$. Moreover, all constraints ~\eqref{cons:even} remain feasible.
\end{proof}

Finally, the competitive ratio is calculated with computer assistance, as provided in Section~\ref{sec:general_upper}.

\section{Failed Attempt for General Vertex Arrival}
\label{app:fail}

In this section, we explain in detail the straightforward adaption of Eager Water-filling to the general vertex arrival setting fails to achieve a competitive ratio better than $0.5$.

Consider the following algorithm with a simple pricing scheme:
Fix an increasing function $f:[0,1] \to [0,1]$.
\begin{itemize}
	\item On the arrival of each vertex $u$, fractionally matches $u$ to the neighbor with the cheapest price $f(x_v)$, as long as $f(x_u) + f(x_v) \le 1$. 
\end{itemize}
When $x_{uv}$ is increased by $dx$ on the arrival of vertex $u$, update $\alpha_u, \alpha_v$ by
\[
d\alpha_u = (1-f(x_v)) dx \quad \text{and} \quad d\alpha_v =f(x_v) dx.
\]
Note that this is a special case of the main algorithm we presented in Section~\ref{sec:general_lower} by setting $g(a_v,x_v) = f(x_v)$ for all $a_v,x_v$.
Therefore, all the analysis applies to the above algorithm. Let $h(\tau) \eqdef f^{-1}(\tau)$.
Adapting Lemma~\ref{lem:opth_general_arrival}, a $\Gamma$-competitive algorithm needs a function $h$ satisfying the following condition:
\[
	\forall 0\le \tau_v \le \theta_v \le 1 \text{ and } \tau_u \ge 1-\theta_v:\quad  \theta_v \cdot h(\theta_v) - \int_{\tau_v}^{\theta_v} h(y_v) dy_v + \tau_u \cdot h(\tau_u) \ge \Gamma~,
\] 
where $\tau_v$ corresponds to the price of $v$ after the arrival of $v$, $\theta_v$ corresponds to the price of $v$ after the arrival of $u$, and $\tau_u$ corresponds to the price of $u$ after the arrival of $u$.

Consider the special case when $\tau_v=0$ and $\tau_u = 1-\theta_v$. The above family of conditions become:
\[
\forall 0\le \theta_v \le 1:\quad  \theta_v \cdot h(\theta_v) - \int_{0}^{\theta_v} h(y_v) dy_v + (1-\theta_v) \cdot h(1-\theta_v) \ge \Gamma~.
\]
Unfortunately, factor revealing lp shows that the largest feasible $\Gamma$ is $0.5$ for all functions of $h$.

\end{document}